\documentclass[atmp]{ipart_v1}

\Vol{19}

\Issue{3}

\Year{2015}

\firstpage{701}

\usepackage{t1enc}

\usepackage[latin1]{inputenc}

\usepackage[english]{babel}

\usepackage{fullpage}

\usepackage{amsmath,amsfonts,amssymb,amsthm,amscd}

\usepackage{yfonts}

\usepackage{bbm}

\usepackage{bm}

\usepackage{mathrsfs}

\newcommand{\be}[0]{\begin{equation}}

\newcommand{\ee}[0]{\end{equation}}

\newcommand{\B}{{\mathcal B}}
\newcommand{\C}{{\mathcal C}}
\newcommand{\E}{{\mathcal E}}
\newcommand{\F}{{\mathcal F}}
\newcommand{\G}{{\mathcal G}}
\newcommand{\cL}{{\mathcal L}}
\newcommand{\N}{{\mathcal N}}

\newcommand{\V}{{\mathcal V}}
\newcommand{\W}{{\mathcal W}}
\newcommand{\Y}{{\mathcal Y}}
\newcommand{\X}{{\mathcal X}}
\newcommand{\Z}{{\mathcal Z}}

\newcommand{\K}{{\mathcal K}}

\newcommand{\VV}{{\mathbb V}}
\newcommand{\KD}{{\mathbb D}}
\newcommand{\KRP}{{\mathbb R} {\mathbb P}}

\newcommand{\KR}{{\mathbb R}}
\newcommand{\KC}{{\mathbb C}}
\newcommand{\KZ}{{\mathbb Z}}
\newcommand{\KN}{{\mathbb N}}
\newcommand{\KT}{{\mathbb T}}
\newcommand{\pt}{\protect{\underline{\mathrm{pt}}}}
\newcommand{\npt}{\protect{\mathrm{pt}}}
\numberwithin{equation}{section}
\theoremstyle{plain}

\newtheorem{theorem}{Theorem}[section]
\newtheorem{lemma}[theorem]{Lemma}

\newtheorem{corollary}[theorem]{Corollary}
\newtheorem{definition}[theorem]{Def}

\begin{document}
\title[Topological T-duality for Stacks]{Topological T-duality for Stacks using a Gysin Sequence}
\author[Ashwin S. Pande]{Ashwin S. Pande\footnote{ashwin.s.pande@gmail.com, aspande@niser.ac.in},\\
School of Mathematical Sciences,\\
NISER, HBNI, Bhubaneshwar, India.}
\begin{abstract}
In this paper we study the topological T-dual of 
spaces with a non-free circle action mainly using the
stack theory method of Bunke and co-workers \cite{Bunke1}.
We first compare three formalisms for obtaining
the Topological T-dual of a semi-free $S^1$-space in a simple example.  Then, 
we calculate the T-dual of general KK-monopole backgrounds using the stack
theory method. We define the dyonic coordinate for these backgrounds.
We introduce an approach to Topological T-duality using classifying spaces
which simultaneously generalizes the methods of Bunke et al \cite{Bunke1} and 
Mathai and Wu \cite{MaWu}. Then, we define a cohomology Gysin sequence
for prinicpal bundles of stacks and describe an application to Topological
T-duality for stacks. We apply the above to calculate the Topological
T-dual of a general compact three-manifold with an {\em arbitrary} smooth circle action. 
We point out a possible application of these T-duals to higher-dimensional black holes. 
\end{abstract}
\maketitle
\section{Introduction}
\label{secIntro}
Topological T-duality is a recent theory inspired by the
theory of T-duality in String Theory.  Principal circle (and torus) 
bundles with a class in $H^3$ of the total space of the bundle possess
an unusual symmetry, namely, from this data it is possible
to naturally construct {\em T-dual} bundles which are also
principal circle bundles over the same base with a class in $H^3$
of the total space of the T-dual bundle
(See Ref.\ \cite{MRCMP, Bunke} for examples). 

In Refs.\ \cite{Bunke1,Bunke2}, the authors generalize
Topological T-duality to principal bundles of topological
stacks (see Refs.\ \cite{Bunke1,Noohi1, Noohi2, Heinloth} for an introduction)
$\E \to \Y$ with an $S^1$-gerbe $\G$ on the stack $\E$. The authors
show that in such a situation, the T-dual 
exists and is also a principal bundle of stacks together with a gerbe on it.
It was argued in Ref.\ \cite{Bunke2} that since the T-dual exists for principal
bundles of stacks,  it should be possible to compute the T-dual of
spaces with circle group actions which are not 
necessarily free. In this paper we study the Topological T-dual of several
classes of spaces with a non-free circle action. 

A brief outline of the paper is as follows:
In Ref.\ \cite{Pande} the T-duals of some semi-free $S^1$-spaces were
derived using $C^{\ast}$-algebraic techniques. In Ref.\ \cite{MaWu}
a general formalism using the Borel construction was used to
derive the T-dual of any semi-free $S^1$-space. Neither of these
constructions used stack theory, and it would be interesting to
compare the T-duals obtained using all three theories. In Sec.\ (\ref{secComp})
below, we attempt to do this in a simple example.

We calculate Topological T-duals of semi-free spaces 
for the examples of Ref. \cite{Pande} (these are the Kaluza-Klein monopole
backgrounds of string theory) using the methods of Ref.\ \cite{Bunke1} 
in Sec.\ (\ref{secTDKK}). We also comment on the results obtained.

A phenomenon seen in backgrounds with Kaluza-Klein monopoles is the dyonic
coordinate (See Refs.\ \cite{Sen,HaJen} for details). 
A model for this using $C^{\ast}$-algebraic methods
was developed in Ref.\ \cite{Pande}. We argue in Sec.\ (\ref{secDyon})
that this phenomenon may also be obtained completely independently
of the $C^{\ast}$-algebraic formalism in the stack theory formalism
of Ref.\ \cite{Bunke1} using the results from Ref.\ \cite{TTDA}. 

Classifying spaces for stacks were introduced in Ref.\ \cite{Noohi1}.
In Sec.\  (\ref{secClass}) we introduce a formalism based on 
classifying spaces for stacks which simultaneously generalizes
the methods of Bunke et al \cite{Bunke} and Mathai and 
Wu \cite{MaWu} to determine the T-dual of a given principal 
bundle of stacks. Using this we prove an interesting property
of  Topological T-duals calculated using any of the above methods.
 
In Ref.\ \cite{Ginot}, the authors give a Gysin
sequence for a $S^1$-stack. Very roughly, this is an
exact sequence of stack cohomology groups which is derived 
by taking classifying spaces for $\E$ and $\Y,$
obtaining an ordinary principal bundle, and then using
the ordinary (homology) Gysin sequence. In Sec.\ (\ref{secGysin})
we derive a cohomology Gysin sequence based on this argument. 
We argue that this may be used to determine the T-dual principal
bundle of stacks just as the ordinary Gysin sequence may be used to determine
the T-dual of a principal bundle of spaces.
In Sec.\ (\ref{secGysin}) we develop this argument in more detail
and calculate a few  T-duals. We prove that the T-dual of a semi-free space obtained 
using this method is the one obtained by Mathai and Wu in Ref. \cite{MaWu}.

In Sec.\ (\ref{secTDKK}) we had determined the T-dual of a three-manifold,
the $KK$-monopole spacetime. In Sec.\ (\ref{sec3Mfd})  we apply the arguments
in Secs.\ (\ref{secClass}, \ref{secGysin}) to
determine the T-dual of a compact three-manifold with
an {\em arbitrary} circle-action.  Many of these spaces are three-manifolds
with a non-free circle action. 

We conclude with a few remarks on this paper in Sec.\ (\ref{secFinal}).

Throughout this paper we restrict ourselves to circle actions on stacks
and stacks which are principal circle bundles over a stack.
We us the formalism and notation of topological stack theory  
developed in Noohi (Refs.\ \cite{NoohiW,Noohi1,Noohi2,Noohi3}) and in
Heinloth (Ref.\ \cite{Heinloth}) for the rest of this paper.
In addition, in the following,  if we use the formalism or notations from
other papers in the stack theory literature in any section, 
we mention those papers there. We have tried to ensure that the usage of
all stack theory notation in this paper is consistent with the notation
of Noohi.

\section{Comparision of the Three Formalisms}
\label{secComp}
We use the notations and definitions given in the review by 
Heinloth in Ref.\ \cite{Heinloth} in
this paper. In particular, if $X$ is a topological
space then $\underline{X}$ is $X$ viewed as a {\em stack} using the Yoneda
Lemma.  Further let $G$ be a topological group acting continuously
on a topological space $X,$ then $[X/G]$ is the quotient stack of the
{\em topological space $X$} by
this action as in Heinloth, Ref. \ \cite{Heinloth}, Example (2.5) or
Noohi, Ref.\ \cite{NoohiW}, Sec. (1) .

Three completely different formalisms have been proposed to calculate the topological
T-dual of a space with a circle action: The formalism of Mathai and Rosenberg based on
continuous-trace algebras and the crossed product
in Ref.\ \cite{MRCMP}, the formalism of Bunke et al using methods from algebraic topology
in Ref.\ \cite{Bunke, Bunke1, Bunke2},
and the formalism of Mathai and Wu using methods from equivariant cohmology in Ref.\ \cite{MaWu}.
 
To prevent confusion in the rest of this paper  
we refer to the formalism of Mathai and Rosenberg in Ref.\ \cite{MRCMP} as 
the {\em $C^{\ast}-$algebraic formalism} of Topological T-duality.
We also refer to the formalism of Bunke et al in Ref.\ \cite{Bunke, Bunke1,Bunke2}
as the {\em stack formalism} of Topological T-duality.
Throughout this paper we will also use the formalism of Mathai and Wu 
in Ref.\ \cite{MaWu} as needed.
In this section we propose to compare the T-duals calculated
using these three formalisms in some simple examples. 

We use the notation of Ref.\ \cite{Heinloth} and refer to 
the stack $\pt$ as the canonical stack associated to the singleton space
$\mbox{pt} = \{ \ast \}$ by the Yoneda lemma.

Consider a space which is a point $\pt$ with an $S^1$-action which
fixes that point. One reason to study this space is that this space is 
extremely simple and yet shows the difference between the three
formalisms. In addition,
this space (see proof of Lemma (\ref{LemStSp}) below) is stack
homotopy equivalent to a more interesting space 
(the cone over the three sphere with a circle action)
which we also study in Sec.\ (\ref{secTDKK}) below.

It is interesting to compare the T-duals obtained
for this space using the formalisms of Refs.\ \cite{MRCMP,Bunke1, MaWu}.
We have the following theorem:
\begin{theorem}
Consider a point $\pt$ with an $S^1$-action which fixes the point.  
Then,
\leavevmode
\begin{enumerate}
\item The Toplogical T-dual of the space $\pt$ with the above
circle action in the formalism of Mathai and Rosenberg (see Ref.\ \cite{MRCMP})
is $\KR$ with quotient space the point.
\item The Topological T-dual of the principal bundle of stacks
$\pt \to [\pt/S^1]$ in the formalism of Bunke et al.\  
(see Refs.\ \cite{Bunke,Bunke1,Bunke2}) is the principal bundle of stacks
$[\npt/S^1] \times S^1 \to [\npt/S^1]$ with
a gerbe on the total space corresponding to $H$-flux.
\item The Topological T-dual of the space $\pt$ in the formalism of Mathai-Wu (see 
Ref.\ \cite{MaWu}) is $ BS^1 \times S^1 \times \npt$ with $H$-flux.
\end{enumerate}
\label{ThmComp}
\end{theorem}
\begin{proof}
Suppose one considers a point with a $S^1$-action which fixes
the point. The quotient is still a point. 
\leavevmode
\begin{enumerate}
\item In the $C^{\ast}$-algebraic formalism, to this geometry 
we would assign the $C^{\ast}$-algebra of compact operators $\K$ 
(since the spectrum of $\K$ is the point). The trivial action of
$S^1$ on $\K$ lifts to a trivial action $\alpha$ of $\KR$ on $\K$.
Also, any other action of $\KR$ on $\K$ is exterior equivalent to
this trivial one.
The T-dual would be the spectrum of the crossed product 
$\K \underset{\alpha}{\rtimes} \KR \simeq C_0(\KR)$ i.\ e.\  $\KR$ 
where $\simeq$ denotes Morita equivalence (See Ref.\ \cite{RaeRos}).
The group action on the T-dual $C^{\ast}$-algebra $C_0(\KR)$ is induced from
the translation action of $\KR$ on itself (See Ref.\ \cite{RaeRos}
for details). The quotient would be the point as expected.
\item 
In Ref.\ \cite{Bunke1}, Sec.\ (4.2, 4.3), Prop.\ (4.3), the T-dual of a stack 
is obtained by the following procedure: One passes to the geometric realization 
of the simplicial space of the groupoid associated with that stack.
This gives an ordinary principal bundle with $H$-flux from the principal
bundle of stacks one one began with. This principal bundle may be T-dualized
in the normal way \cite{Bunke1}.

Consider the principal bundle of stacks $ \pt  \to  
[\npt / S^1] $ 
and the atlas $\mbox{pt} \to [\npt/S^1].$ Let
$Y = \mbox{pt} \underset{[ \pt/S^1]}{\times} \pt $ 
be an atlas for $\pt$ (See Ref.\ \cite{Bunke1} Sec.\ (4.2)).
We have a commutative square
$$
\begin{CD}
             Y  @>>>      \pt  \\
            @VVV          @VVV \\
             \mbox{pt}      @>>>    [\npt / S^1 ].
\end{CD}
$$

The groupoid associated to $ \pt  \to [ \npt / S^1 ]$ is 
$ \mbox{pt}  \times  S^1  \rightrightarrows  \mbox{pt} $  as there is a 
canonical isomorphism $ Y \simeq ( \mbox{pt}  \times  S^1 )$ since
$( \mbox{pt}  \times  S^1 )$ is the canonical bundle over $ \mbox{pt} $  
(See Heinloth Ex. 2.5 and following).
Similarly the groupoid associated
to the atlas $Y= \left( \mbox{pt}  \underset{[ \npt / S^1 ]}{\times}  
\pt \right) \to  \pt $ is $Y \underset{\pt}{\times}{Y} \rightrightarrows  \npt$.
Since the fiber product of $Y$ with itself over $\npt$ 
is $ \mbox{pt}  \times ( S^1 )^2$ the
associated groupoid would be $\mbox{pt}  \times ( S^1 )^2 \to  \mbox{pt} $. 

It is clear that the iterated fiber product of $Y$ with itself
$n$ times would be isomorphic to $\mbox{pt} \times ( S^1 )^n$.
The total space of the associated simplicial bundle would 
then (by definition of $EG$) be $ES^1$ and the base
would (by the construction above) be $BS^1 $.
Therefore the T-dual of the simplicial bundle
would be $BS^1  \times  S^1 $ with a gerbe on total space 
(See Refs.\ \cite{Bunke,Bunke1}).  This corresponds to the
T-dual bundle $[\npt/S^1] \times S^1 \to [\npt/S^1]$ with
a gerbe on the total space of the bundle corresponding to the $H$-flux.
\item In the formalism of Mathai and Wu (See Ref.\ \cite{MaWu}), the original
space would be replaced by $ES^1 \times \npt$ as a principal circle bundle over
$BS^1 \times \npt$ and the T-dual would be 
$BS^1 \times S^1 \times \npt$ as a principal circle
bundle over $BS^1 \times \npt$ with $H$-flux. 
The T-dual obtained here namely $BS^1 \times S^1 \times \npt$ 
should be compared with the T-dual $[\npt/S^1] \times S^1$
obtained in Part (2). 
\end{enumerate}
\end{proof}

Thus, the formalisms of Bunke-Schick and Mathai-Wu give 
similar answers here for this example and the $C^{\ast}$-algebraic
formalism gives a different one. This difference is probably due
to the fact that in Ref.\ \cite{MRCMP}, an $S^1$-action on a principal
bundle lifts to an $\KR$-action (with $\KZ$-stabilizers) on the 
$C^{\ast}$-dynamical system associated 
to that space while in Ref.\ \cite{MaWu}, the $S^1$-action remains
an $S^1$-action. That is, in the $C^{\ast}$-algebraic formalism a circle action
is viewed as an $\KR$-action with $\KZ$-stabilizers while in the other
formalisms  circle action is viewed only as a circle action.

In Topological T-duality, it is expected that the
original and T-dual spaces have the same $K$-theory up to a degree shift. 
It was shown by Mathai and Rosenberg in Ref.\ \cite{MRCMP} using the 
Connes-Thom isomorphism, that the $C^{\ast}$-algebraic
T-dual will have this property. 
Similarly, the Topological T-duals calculated from the other two formalisms 
will also have this property. This was demonstrated by Bunke and co-workeres
in Refs.\ \cite{Bunke, Bunke1, Bunke2} for
the stack theory formalism.  For the formalism of Topological T-duality using
equivariant cohomology proposed by Mathai and Wu this was proved in
in Ref.\ \cite{MaWu}.

\section{T-dual of Kaluza-Klein monopole Backgrounds}
\label{secTDKK}
We now restrict our attention to spaces with semi-free circle actions.
We may further restrict ourselves to spaces which contain Kaluza-Klein
monopoles ($KK$-monopoles)\footnote{ In String Theory, the Taub-NUT metric 
(see Refs.\cite{HaJen}) describes a space with one $KK$-monopole.
For more details see Ref.\ \cite{Pande} and references therein.}.
Away from fixed points such spaces are
equivariantly $S^1$-homeomorphic to the total space of a principal
circle bundle. Since the Topological T-dual of a principal circle bundle is
well-known, it is enough to determine the T-dual of a neighbourhood of
the fixed point (by the discussion before Thm.\ (\ref{ThmTDGen}) and by
Thm.\ (\ref{ThmTDGen}) below).

In a neighbourhood of a fixed point four-manifolds with a semi-free circle
action are equivariantly $S^1$-homeomorphic to $\KR^{4}$ with an orthogonal
$S^1$-action.
The associated topological stacks are the stacks
$[CS^3/\KZ_k].$ In Ref.\ \cite{Pande}, we computed the T-duals of these semi-free
spaces using the $C^{\ast}$-algebraic
approach of Mathai and Rosenberg in Ref.\ \cite{MRCMP}. 

Stacks of the form $E=[CS^3/\KZ_l]$ with
$E/S^1 = CS^2$ for $l \geq 2$, and $E=CS^3$ and $E/S^1=CS^2$ for $l=1$
are the stacks associated to the total spaces of 
$KK$-monopoles of charge $l \in \KN, l > 0.$

In Thm.\ (\ref{ThmKK}) below we compute the T-dual of these spaces.
For $l=1$ the associated principal bundle of stacks would be
$\underline{CS^3} \to [CS^3/S^1]$.
For $l \geq 2$, the associated principal bundle of stacks would be 
$[CS^3/\KZ_l] \to [CS^3/S^1]$. We consider the T-dual of spaces
containing multiple $KK$-monopoles in Thm.\ (\ref{ThmTDGen}) and
Cor.\ (\ref{CorKKMulti}) below.

Consider a spacetime which is a $KK$-monopole spacetime with
charge $1$. This corresponds to T-dualizing the space
$CS^3$ with its natural circle action. Physically, the T-dual would 
have a source of $H$-flux over the set in the base corresponding to the
image of the singular fiber and the $H$-flux would be undefined
at the location of the source\footnote{See Ref.\ \cite{Pande} for
a detailed discussion.}.  
In the $C^{\ast}$-formalism of Topological T-duality
the $C^{\ast}$-algebra describing the background
loses the continuous-trace property exactly on this locus.
In Ref. \cite{Pande}, it was argued that this is a model for
a space with a {\em source} of $H$-flux.

In the formalisms of Refs.\ \cite{Bunke1,MaWu}, however, the
T-dual $H$-flux would be everywhere defined, that is, there would be no
{\em source} of $H$-flux present. In particular, for these two theories,
the following would hold: The T-dual of a single NS5-brane with
a background of $k$-units of 
(sourceless) $H$-flux would be indistinguishable from
the T-dual of a space with $(k+1)$-units of sourceless $H$-flux.
We will see this for the stack theory method 
when we calculate the T-dual
of these spaces in Thms.\ (\ref{ThmKK},\ref{ThmTDGen})
and in Cor.\ (\ref{CorKKMulti}) below. We discuss this matter in more
detail after Cor.\ (\ref{CorESt}) below. 

Before we calculate the T-dual of the above spaces, we need 
some preliminary notation and results. Let $\E \to \Y$ be a 
principal $S^1$-bundle of stacks over $\Y$. 
It is clear from the axioms of a principal $S^1$-bundle
(see Ref.\ \cite{Heinloth} after Remark (2.14)) that the
topological stack $\E$ is a space with a left $S^1$-action 
(in the sense of Ref.\ \cite{Ginot}, Sec.\ (3)).
In Sec.\ (\ref{secComp}) above, we had used the definition of
 a quotient stack $[X/G]$ of a {\em topological space} 
$X$ with an action of a topological group $G.$ 
In this Section we use in addition the notation and
notion of a quotient stack of an {\em arbitrary topological stack}
by the action of a topological group $G$
from the work of Ginot and Noohi,  Ref.\ \cite{Ginot}.  
In Ref.\ \cite{Ginot} Sec.\ (4.3) the authors define a {\em quotient stack} 
$[G \backslash \X]$ of a  group action $G$  on an arbitrary
{\em topological stack} $\X.$ We assume the reader is familiar with
these ideas and will refer to them freely
in what follows.

Since the spaces we study in this paper are not orbispaces 
we need to make a remark about the existence of T-dual
stacks for the spaces under study. 
In this paper, we will use the method of Ref.\ \cite{Bunke1} to 
calculate the T-dual of a stack associated to a semi-free space.
In the method of Ref.\ \cite{Bunke1}, if we restrict ourselves 
to $U(1)$-bundles over such stacks, the associated simplicial bundles, 
being circle bundles, may always be T-dualized. 
There is always a `T-duality diamond' of Ref.\ \cite{Bunke} (see
Diagram (2.14) in Lemma (2.13) of that reference) for the
associated simplicial bundles, since these are only circle bundles.
This gives a diagram of the form Diagram (4.1.2) in
Ref.\ \cite{Bunke2} for the associated {\em stacks}.
Hence, for such a stack, the T-dual exists in the sense of Def.\ (4.1.4)
of Ref.\ \cite{Bunke2}, even if the base is not an orbispace.

In the following, the stacks $\E$ we dualize
are not Seifert fibered spaces
as in Ref.\ \cite{Bunke1}, but, due to the above,
the principal bundles $p:\E \to \B$ may be completed
into a diagram of the form of Diagram (4,1.2) of Ref.\ \cite{Bunke2},
and hence for these stacks the T-dual exists
by Ref.\ \cite{Bunke2}, Def.\ (4.1.4).

We now prove some preliminary results about principal bundles
of stacks and spaces with a circle action which will be useful later.
We use the notion for the quotient stack $[G \backslash \E]$
of a {\em topological stack} $\E$
by the action of a topological group $G$ from Ref.\ \cite{Ginot}
in the following.
\begin{lemma}
\leavevmode
\begin{enumerate}
\item
Let $E$ be a space with a circle action. Then,
$\underline{E} \to [E/S^1]$ is a principal bundle of topological stacks.
\item
Let $p:\E \to \Y$ be a principal $S^1$-bundle of topological stacks.
Then $\E$ is a stack with a left $S^1$-action in the sense
of Ref.\ \cite{Ginot} and $\Y \simeq \left[ S^1 \backslash \E \right].$
\item For a space $E$ with a circle action, 
$[S^1 \backslash \underline{E}] \simeq [E/S^1]$.
\end{enumerate}
\label{StLem1}
\end{lemma}
\begin{proof}
\leavevmode
\begin{enumerate}
\item 
$E$ is a space with a $S^1$-action and satisfies the
conditions for a principal $S^1$-bundle described after Remark (2.14) in 
Ref.\ \cite{Heinloth}. This is equivalent to the definition of a 
principal bundle of stacks using atlases\footnote{See Ref.\ \cite{Heinloth},
Def. (2.11)} by the Claim before Example (2.15)
in Ref.\ \cite{Heinloth}.
\item
That $\E$ is a stack with a left $S^1$-action follows from
the definition of a principal bundle of stacks (See Ref.\ \cite{Heinloth}).

From Ginot and Noohi (Ref.\ \cite{Ginot}), Prop.\ (4.8),
$\left[S^1 \backslash \E \right]$ is a stack.
Let $\tilde{p}:\E \to \left[S^1 \backslash \E \right]$ be the natural 
map defined in Sec.\ (4.1) of Ref.\ \cite{Ginot}. 

Let $T \to \Y$ be an atlas for $\Y$. By definition of a principal bundle
of stacks, (see Ref.\ \cite{Heinloth}), we have a commutative square
\begin{equation}
\begin{CD}
             P  @>>>      \E \\
            @VVV          @VVV \\
             T      @>>>    \Y \label{CDYB}
\end{CD}
\end{equation}
where $P$ is a principal bundle over $T$ and an atlas for $\E$. 
As noted above, we have a natural map $\tilde{p}:\E \to [S^1\backslash \E]$.
By the second part of the proof of Prop.\ (4.8)
of Ref.\ \cite{Ginot}, $P$ is also an atlas for $[S^1\backslash \E]$.
By Ref.\ \cite{Ginot}, Sec.\ (4.1)  
the stack $[S^1 \backslash \E]$ is the
stackification of the prestack $\lfloor S^1 \backslash \E \rfloor$.
Then, we have

\begin{gather}
\lfloor S^1 \backslash \E \rfloor(P) \simeq  S^1 \backslash (\E(P)),
(\mbox{by definition of $\lfloor S^1 \backslash \E \rfloor$, see 
Sec.\ (4.1) of Ref.\ \cite{Ginot}}),\nonumber \\ 
\simeq  \Y(P),(\mbox{By definition, see Sec.\ (4.1) of Ref.\ \cite{Ginot} }),
\nonumber \\
\simeq  \Y(P/S^1) \simeq \Y(T),(\mbox{because the Diagram (\ref{CDYB}) 
commutes}). \nonumber 
\end{gather}

Hence, the stackification of $\Y$ is isomorphic to the
stackification of $\lfloor S^1 \backslash \E \rfloor$. Since $\Y$ is
a stack, this implies that $\Y \simeq [S^1 \backslash \E]$.
\item This follows from Parts (1) and (2) above.
\end{enumerate}
\end{proof}

\begin{lemma}
Let $\X$ be a stack with a $S^1$-action in the sense of Ginot et. al (See 
Ref.\ \cite{Ginot}, Def.\ (3.1)). Let $q:\X \to [S^1 \backslash \X]$ be
the quotient map of Ref.\ \cite{Ginot}, Sec.\ (3.2). Then 
$[S^1 \backslash \X]$ is a topological stack and
$q:\X \to [S^1 \backslash \X]$ is a principal bundle of stacks in the
sense of Ref.\ \cite{Heinloth}.
\label{StLem2}
\end{lemma}
\begin{proof}
By definition, we have an action $\mu$ of $S^1$ on the stack $X$.
By Prop.\ (4.8) of Ref.\ \cite{Ginot}, $[S^1 \backslash \X]$ is also
a topological stack. By Prop.\ (4.7) of the same reference, the map $q$
is representable. It can be checked that the conditions for a principal
bundle of stacks given after Remark (2.14) in Ref.\ \cite{Heinloth}
are satisfied with $act = \mu,p=q.$ 
%
\end{proof}

\begin{corollary}
$p:\E \to \Y$ is a principal bundle of stacks iff $\E$ is a stack with
a left $S^1$-action in the sense of Ref.\ \cite{Ginot} and 
$\Y\simeq [S^1 \backslash \E].$
\label{CorStLem}
\end{corollary}
\begin{proof}
This follows from Lemma (\ref{StLem1}), Part (2), and Lemma (\ref{StLem2})
above.
\end{proof}

In Thm.\ (\ref{ThmComp}) (2) we determined the Topological T-dual of
the principal bundle of stacks $\pt \to [\npt/S^1]$. 
The calculation for the Topological T-dual of $CS^3$ is nearly similar
to that for the Topological T-dual of $\pt \to [\npt/S^1]$ for the following
two reasons: Firstly, the space $CS^3$ is equivariantly homotopy equivalent 
to its vertex $\npt$ (see proof of Lemma (\ref{LemStSp} below). 
Secondly, $CS^3$ and $CS^2$ are contractible and are 
homeomorphic to $\KR^4$ and $\KR^3$ respectively. As are result, any
principal bundle over $CS^3$ is trivial. 
Due to this, we can take the above proof and replace $\pt$ with
$CS^3$ and $[ \npt/S^1 ]$ with $CS^2$ and obtain a working proof. 

\begin{theorem}
\leavevmode
\begin{enumerate}
\item The Topological T-dual of the principal bundle of stacks 
$\underline{CS^3} \to [CS^3/S^1]$ (associated to a $KK$-monopole) 
in the formalism of \cite{Bunke1} is the principal bundle of stacks
$[CS^3/S^1 ]\times S^1 \to [CS^3/S^1] $ 
with a gerbe on the stack $[CS^3/S^1] \times S^1$.
\item  Consider the $S^1$-action on a point which fixes the point.
This gives a principal bundle of stacks $\pt \to [\npt/S^1]$.
Consider the subgroup $\KZ_k \hookrightarrow S^1$ for any natural 
number $k > 1$. Then, $[\npt/\KZ_k] \to [\npt/S^1]$ is also a principal bundle
of stacks. This bundle with no $H$-flux
has as T-dual the principal bundle of stacks
$[\npt/S^1] \times S^1 \to [\npt/S^1]$ with $H$-flux of $k$ units.
\item For $k$ any natural number larger than $1$,
there is a principal bundle of stacks
$[CS^3/\KZ_k] \to [CS^3/S^1]$ 
(corresponding to a $KK$-monopole of charge $k$).
The Topological T-dual of this bundle with no $H$-flux, 
in the formalism of \cite{Bunke1}, is the principal bundle of stacks
$([CS^3/S^1] \times S^1) \to [CS^3/S^1]$ with $k$ units of $H$-flux.
\end{enumerate}
\label{ThmKK}
\end{theorem}
\begin{proof}
\leavevmode
\begin{enumerate}
\item Consider the principal bundle of stacks 
$\underline{CS^3} \to \left[CS^3/S^1 \right]$.
The space $\KR^4 \simeq CS^3$ is an atlas for the stack 
$[CS^3/S^1]$.
Let $Y = CS^3 \underset{[CS^3/S^1]}{\times} \underline{CS^3}$ 
be an atlas for $\underline{CS^3}$ (See Ref.\ \cite{Bunke1} Sec.\ (4.2))
induced by the atlas for the stack $[CS^3/S^1]$.
We have a commutative square
$$
\begin{CD}
             Y  @>>>      \underline{CS^3} \\
            @VVV          @VVV \\
             CS^3      @>>>    [CS^3/S^1].
\end{CD}
$$

The atlas associated to $\underline{CS^3}$ is 
$ CS^3  \times  S^1  \rightarrow \underline{CS^3} $: There is a 
canonical isomorphism $Y \simeq ( CS^3  \times  S^1 )$, since, due to
the contractibility of $CS^3 \simeq \KR^4$, $(CS^3  \times  S^1 )$ is
the canonical bundle over $ CS^3 $ (See Heinloth Ex. 2.5 and following). 

The groupoid associated to the atlas
$Y= \left( CS^3  \underset{[CS^3/S^1]}{\times}  \underline{CS^3} \right) 
\to  \underline{CS^3}$ is 
$Y \underset{\underline{CS^3}}{\times}Y \rightrightarrows CS^3$. The fiber
product $Y \underset{\underline{CS^3}}{\times}Y$ is 
$ CS^3  \times ( S^1 )^2$ by definition. Therefore,
the groupoid is $CS^3 \times (S^1)^2 \rightrightarrows CS^3$. 
It is clear from the definition that the iterated fiber product of 
$Y$ with itself $n$ times is $CS^3 \times (S^1)^n$.

The associated simplicial space in each degree
would be $CS^3 \times (S^1)^n$. The simplicial space is thus the
fiber product of $CS^3 \to CS^2 \sim *$ with $* \times (S^1)^n \to *$.
Therefore by Ref.\ \cite{May}, Cor.\ (11.6), the simplicial
space is the fiber product $(A \underset{*}{\times} ES^1)$ where $A$
is the geometric realisation of the simplicial space which is
$CS^3$ in each degree. Since $CS^3$ is contractible, 
the space $(A \underset{*}{\times} ES^1)$ would be homotopy
equivalent to $ES^1$. 

We had noted above that the total space of the bundle is the simplicial
space which in each degree is $CS^3 \times (S^1)^n$. 
Hence, the base of the simplical bundle would be 
the simplicial space which in each degree $>1$ is $CS^2 \times (S^1)^{n-1}$
and at degree $1$ is $CS^2 \times \mbox{pt}$.
Let $B$ be the simplicial space which is $CS^2$ in each degree.
By the above argument the base is the fiber product
$(B \underset{*}{\times} BS^1)$. Since $B$ is contractible, the base has the 
homotopy type of $BS^1$. Thus, the simplicial bundle associated to
the bundle of stacks $\underline{CS^3} \to [CS^3/S^1]$ is
$(A \underset{*}{\times} ES^1) \to (B \underset{*}{\times} BS^1)$.

Therefore the T-dual of the simplicial bundle
would be $(B \underset{*}{\times}BS^1) \times S^1 $ 
with a gerbe on total space (See Refs.\ \cite{Bunke,Bunke1}).  
The class in $H^3$ of the total space associated to the gerbe
would correspond to the generator of 
$H^3((B \underset{*}{\times} BS^1) \times S^1, \KZ) \simeq
H^3(BS^1 \times S^1, \KZ)$.
It is clear from the above that this is the simplicial bundle 
associated to $[CS^2/S^1] \times S^1$. Also, the $H$-flux on 
the total space of the bundle is $1$. Therefore the
T-dual stack has a gerbe on it.
\item 
Consider the $S^1$-action on a point which fixes the point. 
This gives a $S ^1$-bundle of stacks $\pt \to [\npt/S^1]$.
Then $\KZ_k \hookrightarrow S^1$ acts on $\pt$ as
well. The $S^1$-action on $\pt$ gives an action of $S^1$ 
on $[\npt/\KZ_k]$. The quotient is still $[\npt/S^1].$
Therefore, $[\npt/\KZ_k] \to [\npt/S^1]$ is a principal bundle of stacks
by Lemmas (\ref{StLem1},\ref{StLem2}) above.

The groupoid associated to the stack $[\npt/\KZ_k]$ is 
$\KZ_k \rightrightarrows \pt$ and the associated simplicial space is
$B\KZ_k$ (See Ref.\ \cite{Bunke1} Sec.\ (5.1) ). 
The simplicial space associated to $[\npt/S^1]$ is $BS^1$.
The associated simplicial principal bundle is $B\KZ_k \to BS^1$.
This is $ES^1/\KZ_k \to BS^1$. Thus the T-dual simplicial bundle is
$BS^1 \times S^1$ with $H$-flux $k$. This is the simplicial bundle associated
to the stack $[\npt/S^1] \times S^1$ with a gerbe on the stack.
\item 
For every $k > 1,$ there is an $S^1$-action on $CS^3/\KZ_k$  
induced by the $S^1$-action on the $S^1$-space $CS^3$ with
quotient $CS^3/S^1.$ 
Thus there is an $S^1$-action on $[CS^3/\KZ_k]$ 
such that the quotient under the $S^1$-action is $[CS^3/S^1].$ 
Therefore, by Lemmas (\ref{StLem1},\ref{StLem2}), 
there is a principal bundle of stacks
$[CS^3/\KZ_k] \to [CS^3/S^1]$.

The simplicial space associated to $[CS^3/\KZ_k]$ is
determined in a manner similar to Part (1).
Consider the principal bundle of stacks 
$\underline{CS^3} \to [CS^3/\KZ_k]$.

Let $W$ be an atlas for the stack $\underline{CS^3}$ induced
by the the atlas $CS^3$ for the stack $[CS^3/\KZ_k]$.
Then, $W$ is a principal $\KZ_k$-bundle over $CS^3$.
We have a commutative square
$$
\begin{CD}
             W  @>>>      \underline{CS^3} \\
            @VVV          @VVV \\
             CS^3      @>>>    [\underline{CS^3}/\KZ_k].
\end{CD}
$$
There is a canonical isomorphism $W \simeq (CS^3 \times \KZ_k)$ 
by the same argument as in Part (1) with $S^1$ replaced by $\KZ_k$.
The groupoid associated to the atlas $W$ is 
$W \underset{\underline{CS^3}}{\times} W \rightrightarrows CS^3$.
We have that $W \underset{\underline{CS^3}}{\times} W \simeq CS^3 
\times (\KZ_k)^2$. Also,
$W \underset{\underline{CS^3}}{\times} \ldots 
\underset{\underline{CS^3}}{\times} W \simeq CS^3 \times (\KZ_k)^n$.
Let $A$ be the simplicial space which is $CS^3$ in each degree.
By an argument similar to that in Part (1), 
the simplicial space associated to $\underline{CS^3}$ by this atlas is 
then $(A \underset{*}{\times} E\KZ_k)$. This has a natural 
action of $\KZ_k$ on each factor. It is a principal $\KZ_k$-bundle
over the simplicial space associated to $[CS^3/\KZ_k]$. (By the
above diagram, $W$ is a principal $\KZ_k$ bundle over $CS^3,$
the result follows from the definition of the simplicial space
associated to a groupoid.) Hence, the simplicial space associated
to $[CS^3/\KZ_k]$ by this atlas is $(A/\KZ_k \underset{\ast}{\times} B\KZ_k).$
 It is also a principal circle bundle
over the simplicial space associated to $[CS^3/S^1]$.
\par
By Part (1), the simplicial space associated to $[CS^3/S^1]$
is $(B \underset{*}{\times} BS^1)$. 
Therefore the simplicial circle bundle associated to the principal bundle of
stacks $[CS^3/\KZ_k] \to [CS^3/S^1]$
is $(A/\KZ_k \underset{*}{\times} B\KZ_k) \to (B \underset{*}{\times} BS^1)$. 
\par
Thus, the T-dual would be 
$(B \underset{*}{\times} BS^1) \times S^1 \to (B \underset{*}{\times} BS^1)$
with $H$-flux. This is the principal bundle associated to 
$([CS^3/S^1] \times S^1) \to [CS^3/S^1]$ with a gerbe on it.
\end{enumerate}
\end{proof}

It is interesting to observe that if a property holds for T-duality
pairs in the sense of Ref.\ \cite{Bunke}, it can be applied to 
the pair consisting of the simplicial bundle and $H$-flux. Sometimes, 
this has interesting consequences for spaces with a non-free $S^1$-action.
We present two examples below using this property: In the first we 
calculate the T-dual of a semi-free space with countably many
isolated fixed sets. In the second example done in the next Section, 
we develop a model for the Dyonic coordinate of Ref.\ \cite{Pande} using
the stack-theoretic approach.

Note the following property of T-dual principal $S^1$-bundles:
Let $p:E \to B$ be a principal $S^1$-bundle. Let $h \in H^3(E,\KZ)$ be
an $H$-flux. Let $W_1,\ldots,W_k$ be open
subsets of $B$ such that $W_1 \cup \ldots \cup W_k = B$. Let 
$E_i = p^{-1}(W_i)$ be the induced open cover of $E$. Let $h_i$ be the
$H$-flux restricted to $E_i$. Let $q:E^{\#} \to B$ be the T-dual of 
the principal bundle $E \to B$ with $H$-flux $h^{\#}.$ 
Let $E^{\#}_i=q^{-1}(W_i)$ and let $h^{\#}_i = h^{\#}|_{E_i}$. 
Then we claim that $(E^{\#}_i,h^{\#}_i)$ is the T-dual of $(E_i,h_i).$
This follows from the existence of a classifying 
space for a pair and the properties of the T-duality map on it 
in Ref.\ \cite{Bunke}: If $R$ is the classifying space of pairs of
Ref.\ \cite{Bunke}, $T:R \to R$ the T-duality map and 
$\phi:B \to R$ the classifying map for the
pair $(E,h),$ we have, on restriction to $W_i,$ 
$(T\circ \phi)|_{W_i}=T\circ (\phi|_{W_i}).$ 
However, by definition, $(T \circ \phi)|_{W_i}$ classifies
$(E^{\#}_i,h^{\#}_i)$ while
$\phi|_{W_i}$ classifies $(E_i,h_i)=(E|_{W_i},h|_{E_i}).$

It is natural to ask if the above can be extended to T-dualize
semi-free spaces with many isolated fixed point sets. To do
this we need to be able to glue such spaces to other spaces
with a circle action and to calculate the T-dual.
We show that a result similar to the above holds for semi-free spaces:
\begin{theorem}
\leavevmode
\begin{enumerate}
\item Let $E$ be a semi-free $S^1$-space with at most countably many
isolated fixed sets of the $S^1$-action ${F_1,\ldots,F_k,\ldots}$.
Suppose we are given disjoint neighbourhoods $U_i$ of $F_i$.
Let $\underline{V_i} = [U_i/S^1].$ Then the T-dual of $E$
is determined by the T-dual of the $U_i$ together with
the T-dual of the principal circle bundle $P = (E - \bigcup^{i}_{l=1} U_l).$ 
\item The T-dual of a semi-free space $E$ with at most
countably  isolated fixed sets is
the principal bundle of stacks 
$[E/S^1] \times S^1 \to [E/S^1]$ with $H$-flux.
There will be $H$-flux on the T-dual coming from an NS5-brane if
the T-duals of any of the $U_i$ (see previous part) possess $H$-flux.
\end{enumerate}
\label{ThmTDGen}
\end{theorem}
\begin{proof}
\leavevmode
\begin{enumerate}
\item Let $E$ be a semi-free $S^1$-space with at most countably many
isolated fixed sets of the $S^1$-action ${F_1,\ldots,F_k,\ldots}$. 
Let $W = E/S^1$ and let $U_l,l=1,\ldots,k,\ldots$ be open disjoint 
subsets of $E$, such that each $F_l \subsetneq  U_l$. 
Since the fixed sets are isolated, we may always assume that
$U_l \cap U_k = \phi$ for every $l \neq k$.
Then, by the classification theorem for spaces with finitely many orbit
types (see proof of Cor.\ (\ref{CorKKMulti})), 
$P=(E - \bigcup^{k}_{l=1} U_l)$ is a principal
circle bundle over $V=(W - \bigcup^{k}_{l=1}V_l)$ where $V_l=(U_l/S^1).$ 
Also, $E$ is determined by $P$ and the gluing data for the $U_l$.
Since there is no $H$-flux on $E,$ these data determine the T-dual of $E$
by Thm.\ (\ref{ThmEAttach}) item (4) below.

\item Given the data of the previous part, for every $i,$
suppose we are given atlases $V_i$ for $\underline{V_i}$ and 
induced atlases $Q_i$ for $\underline{U_i}$ in the sense of 
Ref.\ \cite{Bunke1}. Let $SV$ be the simplical space associated to $V$, 
and, for every $i,$ $SV_i$ the simplical space associated to $V_i$ 
by the above atlases. Similarly, let $SP$ be the simplicial space
associated to $P$ and, for every $i,$ $SU_i$ the simplicial space 
associated to $U_i$ by the above atlases. 

Consider the principal bundle of stacks $\underline{E} \to \underline{W}$.
Consider the atlas $X=V \cup_{i} V_i$ for $W$. 
Then $V_i \cap V_j = \phi$ for all $i\neq j$.
Also, $V_i \cap V$ need not be empty, but,
$V_i \cap V \subseteq V_i.$
Now $X \underset{\underline{W}}{\times} X \simeq V \cup V_i \cup (V \cap V_i).$
Similarly, for the same reason, the $n$-fold fiber product 
$X \underset{\underline{W}}{\times} \ldots 
\underset{\underline{W}}{\times}X \simeq V \cup_{i}V_i \cup (V \cap V_i).$ 
However, since $(V \cap V_i) \subseteq V_i,$
$X \underset{\underline{W}}{\times} X$ may always be written
as $V \cup_{i} V_i.$ 
Also, in the associated simplicial space, $V$ is always glued
to each $V_i$ while $V_i$ glue to themselves. Then, the simplicial
space associated to $X$ is $SV \cup_f SV_1 \cup_{g_1} \ldots \cup_{g_k} SV_k \cup \ldots $ for some gluing maps $f,g_i$.

Consider the atlas $Y=P \cup_{i} Q_i$ for $\underline{E}$.
Here also, we have that $Q_i \cap Q_j = \phi$ for all $i\neq j$.  
Also, $Q \cap Q_i \subseteq Q_i$ for every $i.$
This implies that $Y \underset{\underline{E}}{\times} Y$ may 
always be written as $P \cup_{i} Q_i$ 
by the intersection property of $P$ and $Q_i$ described above.
Similarly the $n$-fold fiber product $Y \underset{\underline{E}}{\times} \ldots 
\underset{\underline{E}}{\times}Y$ may always be written as
$P \cup_{i} Q_i$ by the intersection property described above.
Also, in the associated simplicial space, $P$ is always glued
to each $P_i$ while the $P_i$ glue to themselves. Then, the simplicial
space associated to $Y$ is 
$SP \cup_{f'} SQ_1 \cup_{g_1'} \ldots \cup_{g_k'} SQ_k \ldots $ 
for some gluing maps $f',g_i'$.

Therefore we have the associated principal bundle 
$(SP \cup_{f'} SQ_1 \cup_{g_1'} \ldots \cup_{g_k'} SQ_k \ldots)
\to (SV \cup_f SV_1 \cup_{g_1} \ldots \cup_{g_k} SV_k \ldots)$ where $f',f,
g_i', g$ are defined above.

By the remark before this theorem, the T-dual will be 
$$E^{\#}= (\left( SP \cup_f SQ_1 \cup_{g_1}
\ldots \cup_{g_k} SQ_k \ldots \right) \times S^1)$$ as a principal bundle over 
$$B^{\#} = \left( SV \cup_f SV_1 \cup_{g_1} \ldots \cup_{g_k} SV_k \ldots 
\right).$$

Note that this is the principal simplicial bundle associated to
$([E/S^1] \times S^1) \to [E/S^1].$
There will be nonzero $H$-flux on $E^{\#}$ due to the
fact that the original bundle $E$ had nontrivial topology. There will
be additional $H$-flux on $E^{\#}$ due to NS5-branes if there is 
nonzero $H$-flux on the T-dual of any of the bundles $SQ_i$ when T-dualized by
themselves: By the remark before this Theorem, the $H$-flux on the
total space of the simplicial bundle associated to $E^{\#}$
must restrict to this $H$-flux on the subspace $SV_i \times S^1$. 
\end{enumerate}
\end{proof}

This Theorem lets us determine the T-dual of any semi-free space
with countably many isolated fixed sets. This covers most of the
semi-free spaces that would occur in a physical context.
In particular, we may now determine the T-dual of a space with at 
most countably many isolated $KK$-monopoles.
\begin{corollary}
Let $E$ be a semi-free $S^1$-space with at most countably many
Kaluza-Klein monopoles ${p_1,\ldots,p_k,\ldots}$.
Then, the T-dual is a trivial principal bundle 
glued to spaces of the form $([CS^3/S^1] \times S^1)$. 
There is $H$-flux present on the T-dual.
\label{CorKKMulti}
\end{corollary}
\begin{proof}
This is an elementary application of Thm.\ (\ref{ThmTDGen}).
In $E$, since the $KK$-monopoles which are the fixed points of 
the $S^1$-action are isolated, it is possible to
enclose each one in an open set homeomorphic to a ball $CS^3$.
Thus, as topological spaces, each $U_i$ is equivariantly 
homeomorphic to $CS^3$. The atlases $U_i$ and $V_i$ may be chosen 
as in Thm.\ (\ref{ThmKK}).
This construction is always possible by the classification theorem 
for spaces with finitely many orbit types  
since there are only two orbit types (fixed points and free orbits) 
and the fixed points are at most countably many and isolated 
(See Ref.\ \cite{Bredon} Chap.\ V Sec.\ (5)).

Given this, the T-dual may be found. 
Note that $SV_i$ are simplicial bundles associated to spaces of
the form $[CS^3/S^1] \times S^1$ with $H$-flux.
Thus, the T-dual of $E$ is a stack which is a trivial principal bundle 
glued to stacks of the form $([CS^3/S^1] \times S^1)$. 

There is $H$-flux present on the T-dual.
There will be nonzero $H$-flux on the T-dual bundle due to the
fact that the original bundle had nontrivial topology. There will
be additional $H$-flux on this bundle due to NS5-branes if there is 
nonzero $H$-flux on the T-dual of any of the $U_i$.
This is because these will then contribute to a nonzero $H$-flux 
on the associated simplicial bundle: By Thm.\ (\ref{ThmKK}), 
there is nonzero $H$-flux on the T-dual 
of any of the $SU_i \subseteq E$, i.\ e.\, there is a $H$-flux on 
$(SV_i \times S^1) \subseteq E^{\#}$. By the remark before 
Thm.\ (\ref{ThmTDGen}), the $H$-flux on $(SV_i \times S^1)$ is 
the restriction of the $H$-flux on the T-dual to $SV_i \times S^1$.
However, by the above, this restriction is nonzero hence the 
T-dual $H$-flux cannot be zero.
\end{proof}

Note that unlike the $C^{\ast}$-algebraic case (\cite{Pande})
there is no {\em source} of $H$-flux on the T-dual. However, the
T-dual does possess $H$-flux. 
We make this precise in the following:

\begin{lemma}
The stack $[CS^3/S^1]$ is not equivalent
to stack $[\underline{CS^2}]$ associated to
the space $CS^2.$
\label{LemStSp}
\end{lemma}
\begin{proof}
The space $CS^2$ is the coarse moduli space
of the stack $[CS^3/S^1]$ (see below) and, as spaces, $CS^3/S^1 \simeq CS^2$. 

We show that the stack $\underline{CS^2}$ has different
stack cohomology groups to the stack $[CS^3/S^1]$ hence they cannot
be equivalent.

The stack $\underline{CS^2}$ is the stack associated to the
contractible space $CS^2$  and hence its stack cohomology
groups with $\KZ-$coefficients are $H^0 = \KZ, H^i = 0, i >0.$

The stack $[CS^3/S^1]$ is homotopy equivalent to the stack
$[*/S^1] \simeq \underline \B S^1.$ This can be seen by considering
the equivariant inclusion $ \ast \hookrightarrow CS^3$ which includes the 
vertex of the cone into the cone. This gives rise to an inclusion
of stacks (by Ref.\ \cite{Ginot} ??check)
$i:[*/S^1] \hookrightarrow [CS^3/S^1].$

Define an $S^1$-equivariant homotopy between the identity $1:CS^3 \to CS^3$ and 
the projection map $CS^3 \to *$ by, for example, $H:CS^3 \times I \to CS^3$ where the map
$H((p,t), s) = (p,\phi(t-s))$, where 
$\phi(x) = x$ if $x$ is positive and $ \phi(x) = 0$ if $x$ is zero or negative.

It is clear that this is $S^1-$equivariant and descends to a homotopy
$H: [CS^3/S^1] \times I \to [CS^3/S^1].$ This homotopy is a homotopy
equivalence between $[CS^3/S^1]$ and $[*/S^1].$ 
Hence, the cohomology
of $[CS^3/S^1]$ is the same as the cohomology of $[*/S^1]$ by
Sec.\ (17) of Ref.\ \cite{Noohi1} and Sec.\ (11) of Ref.\ \cite{Noohi3}.
\end{proof}
\begin{corollary}
\leavevmode
\begin{enumerate}
\item The T-dual $E^{\#}$ of a semi-free $S^1$-space $E$ with at most
countably many $KK$-monopoles is the principal bundle of stacks 
$[E/S^1] \times S^1 \to [E/S^1]$ with $H$-flux.
\item $E^{\#}$ is a topological stack which is not equivalent to a
topological space if and only if the $S^1$-action on
$E$ has fixed sets.  
\item The natural map $\phi_{mod}:\underline{E^{\#}} \to {\underline{E}}^{\#}_{mod}$ 
is an equivalence iff the $S^1$-action on $\underline{E}$ has no
fixed sets.
\end{enumerate}
\label{CorESt}
\end{corollary}
\begin{proof}
\leavevmode
\begin{enumerate}
\item It follows from the proof of Thm.\ (\ref{ThmTDGen}) that the simplicial
bundle associated to the T-dual stack is the trivial bundle over the base
with $H$-flux. This is the simplicial bundle associated to the principal
bundle of stacks 
$[\underline{E}/S^1] \times S^1 \to [\underline{E}/S^1]$ with $H$-flux.
\item 
First note the following: If $\underline{X}$ is a stack equivalent
to a topological space $X,$ then, restricting the equivalence to a
substack shows that every substack of $\underline{X}$ is
equivalent to a topological space.
Suppose the action had no fixed sets, i.\ e.\, none of the $U_i$ was
present in $E$, then, from the proof of Thm.\ (\ref{ThmTDGen}) 
$E$ would be a topological space and so would $E^{\#}$. 
Now suppose the $S^1$-action on $E$ had fixed sets. Then
one of the $U_i$ would be present in $E$, then, from the same proof, 
the T-dual would contain substacks of the form $[U_i/S^1] \times S^1$.  
Here, by the classification theorem for spaces with finitely many
orbit types (see proof of Cor.\ (\ref{CorKKMulti}) 
and by the proof of Thms.\ (\ref{ThmKK},\ref{ThmTDGen})), 
each of these would be equivalent to
the stack $[CS^3/S^1] \times S^1$.
By Lemma (\ref{LemStSp}) above, this stack is not equivalent
to a topological space.
As a result, the T-dual could not be equivalent to a topological space.

\item Suppose the $S^1$-action on $E$ had fixed points and the
map $\phi_{mod}:\underline{E^{\#}} \to \underline{E^{\#}}_{mod}$ 
induced by $\phi_{mod}$ was an equivalence of stacks.
Then, by the proof of the previous part choosing
suitable neighbourhoods of the fixed points will give an
inclusion of stacks $[CS^3/S^1] \to \underline{E^{\#}}.$
Composing with $\phi_{mod}$ would imply that the map 
$[CS^3/S^1] \times S^1 \to CS^2 \times S^1$
would be an equivalence of stacks. 
Since by Lemma (\ref{LemStSp}) above,
the stacks $CS^2$ and $[CS^3/S^1]$ are {\em not} equivalent, 
the stacks $CS^2 \times S^1$ and $[CS^3/S^1] \times S^1$ aren't
equivalent either. Thus the map $\phi_{mod}$ can't be an equivalence of
stacks.

Conversely, suppose the $S^1$-action on $E$ had no fixed points. Then, by the
previous part of the theorem, the T-dual stack would be equivalent to a space
and so $\phi_{mod}$ would give an equivalence
$\overline{\phi_{mod}}:\underline{E^{\#}} \to {\underline{E}}^{\#}_{mod}.$
\end{enumerate}
\end{proof}

Consider the T-dual $[CS^3/S^1] \times S^1]$ 
of $\underline{CS^3}$: The coarse moduli
space of $[CS^3/S^1]$ is $CS^2$ (See Ref.\ \cite{Noohi1} 
Example (4.13), 
$[CS^3/S^1]$ is the quotient stack of the transformation groupoid 
$((CS^3 \times S^1)\rightrightarrows CS^3)$). 
However, since the topological space $CS^2$ is {\em contractible}
$H^3(CS^2 \times S^1,\KZ)=0$, so there can be no
$H$-flux on $CS^2 \times S^1.$
By the above, however, the {\em stack} $[CS^3/S^1] \times S^1$
possesses $H$-flux. This is because the simplicial space associated
to $[CS^3/S^1]$ (see proof of Thm.\ (\ref{ThmKK})) is nontrivial and homotopy
equivalent to $BS^1.$ Hence,  the {\em stack} cohomology group
$H^3([CS^3/S^1] \times S^1,\KZ)$ is nontrivial. 

This nontrivial {\em stack} cohomology in degree three corresponds, by
the above, to a nontrivial $H-$flux on the T-dual $[CS^3/S^1] \times S^1.$ 
Since the $H$-flux on the T-dual stack would vanish (see Cor.\ \ref{CorKKMulti})
if there were no fixed points of the $S^1$-action on the original space, 
{\em presumably this $H$-flux is the flux generated by the T-dual
NS5-brane}. Note that this also happens for the T-dual of
$[CS^3/\KZ_k]$ for $k>1$ since the T-duals are the 
same as the case above only the $H$-flux changes.

This should also happen in the example in Cor.\ (\ref{CorKKMulti}) above:
The T-dual is a principal bundle $P$ 
glued to copies of $([CS^3/S^1] \times S^1)$.
By the proof of Ref.\ (\cite{Bunke}), the T-dual is a topological stack.
As a space, the coarse moduli space of the T-dual will be 
$P \times S^1$ glued to $CS^2 \times S^1$. Also, by Cor.\ (\ref{CorESt})
the T-dual coarse moduli space will be a trivial principal circle bundle. 
The $CS^2$ factor is contractible and the resulting space 
cannot have nonzero $H$-flux coming from an NS5-brane. (The space
will have $H$-flux only due to the $H$-flux on $P^{\#}$).
However, the T-dual {\em stack} does have $H$-flux coming from this source.

Note that in all these T-duals (see also Cor.\ (\ref{CorESt})) above, 
the reason the T-dual has a nontrivial $H$-flux is due to the fact
that the stack cohomology groups of $[CS^3/S^1]$ are {\em different}
from those of the coarse moduli space $\underline{CS^2}.$ 

Also, it is interesting to note that the physical T-dual spacetime is the
coarse moduli space of the stack.  It would be interesting to
see whether this is true in other examples of Topological T-duality.
We calculate a few more examples of T-duals of three-manifolds
in Sec.\ (\ref{sec3Mfd}) below. 

\section{The Dyonic Coordinate}
\label{secDyon}
In String Theory backgrounds which contain $KK$-monopoles possess a
dyonic coordinate. (See Ref.\ \cite{Sen} for details. 
See also Ref.\ \cite{HaJen}). 
Roughly speaking, a large gauge transformation of the $B$-field on
a $KK$-monopole background under T-duality corresponds to a rotation
of the T-dual NS5-brane around its circle fiber. 
A model for this was constructed for $KK$-monopole backgrounds 
using $C^{\ast}$-algebraic methods in Ref.\ \cite{Pande}.

Large gauge transformations of a gerbe on a space $X$ are given 
by a class in $H^2(X,\KZ)$ (See Ref.\ \cite{TTDA}).
We would like to understand the behaviour of these classes under
Topological T-duality for {\it semi-free} spaces. As we have argued
earlier, the $S^1$-spaces underlying $KK$-monopole spacetimes 
are semi-free spaces.

Using the results of Ref.\ \cite{TTDA}, we show below that
for these semi-free spaces $X$, an automorphism of a trivial gerbe
on $\underline{X}$ gives a class in $H^2(X^{\#},\KZ)$ under 
Topological T-duality.

\begin{theorem}
\leavevmode
\begin{enumerate}
\item Consider the principal bundle of stacks $[\npt/\KZ_k] \to [\npt/S^1]$ for
$k=2$. Consider a trivial gerbe on this stack.
Each cyclic subgroup of the group of automorphisms of the gerbe on 
$[\npt/\KZ_k]$ gives rise to a cyclic subgroup of 
$H^2([\npt/S^1]\times S^1,\KZ)$.
For $k=2$, this may be calculated explicitly.
\item Consider the principal bundle of stacks 
corresponding to a $KK$-monopole of charge $k$. Consider a trivial gerbe on 
the total space of the principal bundle.  Each cyclic subgroup of 
automorphisms of the trivial gerbe on the $KK$-monopole of charge $k>1$
gives rise to a cyclic subgroup of the (second) cohomology of the T-dual 
$H^2([CS^3/S^1] \times S^1,\KZ)$.
\end{enumerate}
\end{theorem}
\begin{proof}
\leavevmode
\begin{enumerate}
\item
Consider a cyclic subgroup of the group of automorphisms of 
the trivial gerbe on $[\npt/\KZ_k]$. It is enough to prove the result 
for the generator of this subgroup. 
An automorphism of the trivial gerbe on $[\npt/\KZ_k]$ gives rise to 
a class in $H^2([\npt/\KZ_k],\KZ)$. 
Consider the proof of Part (2) of Thm.\ (\ref{ThmKK}). 
The simplicial bundle associated to the principal bundle of stacks
$[\npt/\KZ_k] \to [\npt/S^1]$ is $p:B\KZ_k \to BS^1$. 
Since we have a class in $H^2([\npt/\KZ_k],\KZ)$, we obtain a
cohomology class on the simplicial space associated to this stack 
(See Ref.\ \cite{Heinloth}, the proof of Prop.\ (4.7)). In turn this
gives a cohomology class on its geometric realization $B\KZ_k$ in 
$H^2(B\KZ_k,\KZ)$. 

By the argument in Ref.\ \cite{TTDA}, 
Thm.\ (6.3), this class gives rise to a natural 
class in the second cohomology group of the T-dual bundle 
$q:BS^1 \times S^1 \to  BS^1$. 
For all natural numbers $k > 1$, $H^2(B\KZ_k,\KZ) \simeq \KZ/k$.
Under T-duality an element of $H^2(B\KZ_k,\KZ)$ gives a class of the form 
$kq^{\ast}(a) \simeq k (a \times 1) \in H^2(BS^1 \times S^1,\KZ)$ 
for some unknown integer $k$ (where $a$ is the generator
of $H^2(BS^1,\KZ)$ (See Ref.\ \cite{TTDA}, Thm.\ (6.3)). 
Thus, an automorphism of the gerbe on $[\npt/\KZ_k]$
with this characteristic class gives rise to a 
cohomology class on the T-dual stack $[\npt/S^1] \times S^1$.
\item
This part of the proof is very similar to the previous part.
Consider a cyclic subgroup of the group of automorphisms of 
the trivial gerbe on $[CS^3/\KZ_k]$. It is enough to prove
the result for the generator of this subgroup. 
The proof is similar to the proof for the previous part, with
the principal bundle changed. 
An automorphism of the trivial gerbe on $[CS^3/\KZ_k]$
gives rise to a class in $H^2([CS^3/\KZ_k],\KZ)$. 
Consider the proof of Part (3) of Thm.\ (\ref{ThmKK}). 
The simplicial bundle associated to the principal bundle of stacks
$[CS^3/\KZ_k] \to [CS^3/S^1]$ is 
$(A \underset{*}{\times} B\KZ_k) \to (B \underset{*}{\times} BS^1)$ 
(where $A$ and $B$ are defined in Thm.\ (\ref{ThmKK}).
Let $P_k= (A \underset{*}{\times} B\KZ_k)$ and 
$W= (B \underset{*}{\times} BS^1)$ 

Since we have a class in $H^2([CS^3/\KZ_k],\KZ)$, this gives a
cohomology class on the simplicial space associated to this stack 
(See Ref.\ \cite{Heinloth}, the proof of Prop.\ (4.7)) and this, in turn,
gives a cohomology class on its geometric realization, that is, a class in 
$H^2(P_k, \KZ)$. The T-dual bundle is 
$W \times S^1 \to W$ with $H$-flux, and, by Ref.\ \cite{TTDA},
a class in $H^2(P_k,\KZ)$ under Topological T-duality naturally
gives rise to a class in $H^2(W \times S^1, \KZ)$. 
\end{enumerate}
\end{proof}

Thus a class in $H^2(X,\KZ)$ naturally gives rise to a class in 
$H^2(X^{\#},\KZ)$. For a $KK$-monopole background, 
a large gauge transformation of the $B$-field gives rise to a
class in $H^2(X,\KZ)$. By the above this induces a class in $H^2(X^{\#},\KZ)$.
For the analogy with the dyonic coordinate to be complete, the induced 
class in $H^2(X^{\#},\KZ)$ should be viewed as an 
automorphism of the T-dual semi-free space $X_k^{\#}$ which
rotates each fiber through $2\pi$. However, it is not clear how to prove 
this in the stack picture. 

A similar construction was made in the $C^{\ast}$-algebraic picture of
topological T-duality in Ref.\ \cite{Pande} where such a rotation did 
correspond to a {\em nontrivial} {\em spectrum-fixing} automorphism
of the T-dual $C^{\ast}$-algebra: The T-dual automorphism obtained
there was of this type.

\section{Topological T-duality using Classifying Spaces}
\label{secClass}
In this section we prove that the stack theory method of T-duality of Bunke and
coworkers (see Ref.\ \cite{Bunke1, Bunke2}) and the method of T-duality
of Mathai and Wu using equivariant cohomology are connected by
the notion of a {\em classifying space of stacks}. In addition we
demonstrate that if both methods can be applied to a given space,
both will give the same T-dual. The perspective gained by using
classifying spaces will help in later sections when we try to T-dualize
three-manifolds with an arbitrary circle action.

In Ref.\ \cite{Bunke2} item (1.2.3), Bunke and 
coworkers argued that the method of Ref.\ \cite{Bunke1} 
may be generalized to arbitrary $S^1$-spaces: Briefly, given
a space $Y$ with an arbitrary $S^1$-action, one may consder
the quotient map $q: Y \to V$ where $V \simeq Y/S^1.$ It is
usually impossible to argue about the T-dual of $Y$ since $V$
may be very singular if $Y$ is not a principal circle bundle.

By passing to the stacks we may replace $Y$ with the associated
stack $\Y \equiv \underline{Y}$ using the Yoneda lemma.
In addition, we may replace the quotient 
by the natural map of stacks $q: \Y \to \V \simeq [Y/S^1].$
Here the base $Y$ which may be singular has been replaced by
the quotient stack $[Y/S^1].$
 
Even though the original
map of spaces $q$ need not be a principal $S^1$-bundle of
{\em spaces} (since the $S^1$-action need not be free), 
the map of {\em stacks} $q: \Y \to \V$ is a principal 
$S^1$-bundle of {\em stacks}.
This principal circle bundle of stacks
may be T-dualized using the arguments given by Bunke
and coworkers in Ref.\ \cite{Bunke1} 
and yields a stack as a T-dual for the space $Y.$
 (We have also discussed this method of Bunke and coworkers in 
Sec.\ (\ref{secComp}) above.)

Another way to T-dualize arbitrary $S^1$-spaces was 
proposed by Mathai and Wu in Ref.\ \cite{MaWu}:
To a space with an arbitrary $S^1$-action $Y,$ we
may associate a space $Y \times ES^1$ with
the product $S^1$-action. The $S^1$-action 
on this space is free even if the action on $Y$ is
not. This space is actually a principal bundle over the
space $(Y \times ES^1)/S^1.$ The latter space is called
the Borel construction.

To a space $Y$ with an arbitrary $S^1$-action and $H$-flux
$H,$ Mathai and
Wu argue that there is associated a nonsingular correspondence space
$\hat{Y}$ which is an equivariant circle bundle over $Y$ 
such that there is a commutative diagram of spaces:
\begin{gather}
\begin{CD}
Y @<<< \hat{Y} \\
@VVV                 @VVV \\
Y/S^1 @<<< \hat{Y}/S^1.
\end{CD}
\label{MaWuSing}
\end{gather}
Here the quotients $Y/S^1$ and $\hat{Y}/S^1$ might
be singular.
The authors argue that this diagram may be replaced by another
diagram containing nonsingular spaces obtained from the Borel
construction above:
\begin{gather}
\begin{CD}
Y \times ES^1 @<<< \hat{Y} \times ES^1 \\
@VVV                                     @VVV \\
Y_{S^1} @<<<      \hat{Y}_{S^1}
\end{CD}
\label{MaWuNS}
\end{gather}
where $Y_{S^1} \equiv Y \times_{S^1} ES^1 \simeq (Y \times ES^1)/S^1$ (and
similarly for $\hat{Y}_S^1$).

In Thm.\ (1) in Ref.\ \cite{MaWu}, the authors argue that in the situation
of Eq.\ (\ref{MaWuSing}), the Topological T-dual of $Y$ (with $H$-flux $H$) is
the (possibly singular) space $\hat{Y}/S^1$.  In the case when
this space is singular, the authors use the term
`Topological T-dual' to refer to the (nonsingular) correspondence space
$\hat{Y}$ - see Ref.\ \cite{MaWu}, Thm.\ (1) and following.  If the T-dual
exists, the twisted equivariant cohomology of $\hat{Y}$ gives the twisted
cohomology of the T-dual. When the T-dual is singular, the twisted
equivariant cohomology of $\hat{Y}$ may still be used, that is, from Ref.\ \cite{MaWu}:
$$H^{\bullet +1}(\hat{Y}_{S^1}, \hat{H}) \simeq H^{\bullet +1}_{S^1}(\hat{Y},\hat{H})$$
where $H^{\ast}(X,H)$ is the twisted cohomology of $X$ twisted by the class $H \in H^3(X,\KZ)$
and $H^{\ast}_{S^1}(X,H)$ is the twisted equivariant cohomology of $X$ twisted by the
class $H \in H^3(X,\KZ)$ viewed as a three-form $H-$flux $H$ on $X.$  

Similarly, the space $\hat{Y}/{S^1}$ may be recovered
from the nonsingular space $\hat{Y}_{S^1}$ above when
the quotient exists. For the purposes of calculation, one may work with the
nonsingular diagram Eq.\ (\ref{MaWuNS}) above.

The nonsingular diagram in Eq.\ (\ref{MaWuNS}) above is actually a Topological T-duality
diamond in the sense of Ref.\ \cite{MRCMP} since  each of the fibrations
$Y \times ES^1 \to Y_{S^1}$ and $\hat{Y}_{S^1} \to Y_{S^1}$ are actually 
principal circle bundles. The principal circle bundle $\hat{Y}_{S^1} \to Y_{S^1}$
may be identified with the T-dual  of the principal circle bundle $Y \times ES^1 \to Y_{S^1}.$

Thus, the Topological T-dual of an arbitrary $S^1$-space $Y$ proposed by Mathai and Wu 
may be obtained from the principal bundle of spaces $Y \times ES^1 \to Y_{S^1}$ and 
the ordinary Topological T-dual (in the sense of Ref.\ \cite{Bunke}) 
of this principal circle bundle. 

We argue here that both these seemingly different
prescriptions for Topological T-duality can be understood from the 
notion of a {\em classifying space} associated to a 
stack.

The following Theorem is Thm. (2.1) from 
Ref.\ \cite{BGNX}, Sec. (2.2).):
\begin{theorem}{\null}
For every topological stack $\X$, there exists a topological space
$X$  together with a morphism $\phi:X \to \X$ which has the
property that, for every morphism $T \to \X$ from a topological
space T, the pullback $T \times_{\X} X \to T$ is a weak homotopy
equivalence.
\end{theorem}

A topological space with the above property is termed as a {\bf classifying space}
for $\X$.  By Ref.\ \cite{Noohi3}, every stack $\V$ possesses a classifying space 
$\Theta(\V)$ which is an atlas $\phi:\Theta(\V) \to \V$ possessing
the above property.

There are at least two classifying spaces naturally associated to a given stack,
the Haefliger-Milnor classifying space and the simplicial classifying space 
\footnote{See Ref.\ \cite{Noohi3} especially Secs.\ (4), (5) and (6).}.
(By Ref.\ \cite{Noohi3}, all classifying spaces for the same stack 
are weak homotopy equivalent,
hence the classifying space is actually a
functor $\Theta: {\mathbf {Stacks}} \to ({\mathbf{ Top}})_{w.e.}.$ 
Here $({\mathbf{Top}})_{w.e.}$ is
the homotopy category of topological spaces, i.e. the
category of topological spaces localized at the weak
homotopy equivalences.)

Before we study Topological T-duals, we need to discuss
some facts about classifying spaces and principal bundles
of stacks.

Let $p: \Y \to \V$ be a principal circle bundle of stacks.
To determine the Topological T-dual of $\Y,$ we need
a gerbe $\G$ on $\Y$ for the $H-$flux.

We pull back the classifying space of $\V,$
$V = \Theta(\V)$ along $p$ to a classifying space $Y$ for the stack
$\Y,$ that is, we  pull back $\V$ along $p$ to $\Theta(\Y).$
Hence, we obtain the following ($2$-commutative) diagram
$$
\begin{CD}
Y \simeq V \times_{\V} \Y @>{p^{\ast}(\phi)}>> \Y \\
@VVV      @VVpV \\
V = \Theta(\V) @>>\phi> \V.
\end{CD}
$$
This may always be done since $p$ is a principal bundle of stacks. 
It can be shown (see Ref.\ \cite{Ginot}, proof of Thm.\ (8.4)) 
that $Y \to V$ is a principal circle bundle of {\em spaces}
if $p$ is a principal circle bundle of {\em stacks}.

The above presumes a {\em natural} way to choose
the classifying space $\Theta(\V)$ for a given stack $\V.$
Both the approaches to Topological T-duality for stacks
require a choice of a classifying space $\Theta(\V)$ not
just a weak equivalance class of classifying spaces.
There are at least two classifying spaces naturally 
associated to a given stack $\V,$  the Haefliger-Milnor
classifying space and the Simplicial classifying space (see 
Ref.\ \cite{Noohi3} Sec.\ (4)).

The Topological T-duality for stacks of Bunke et al uses the 
Simplicial classifying space for the stack $\V.$
We consider the $S ^1$-space $Y$ as a {\em stack} $\Y$
with an $S^1$-action. Similarly we consider the stack 
$\V \simeq [Y/S^1]$ and the natural map $q:\Y \to \V$ as
in Sec.\ (\ref{secTDKK}) above. 

If we have a groupoid presentation
$[Y/R]$ for the stack $\Y$ and a similar presentation 
$[V/R]$ for the stack $\V,$ we may use the Simplicial classifying
spaces $|\Y|$ and $|\V|$ associated to $\Y$ and $\V$ respectively
(see Refs.\ \cite{Noohi3, Bunke1}). We obtain the natural (2-commutative) square:
$$
\begin{CD}
|\Y| @>>> [Y/R] \simeq \Y \\
@VVV          @VVpV \\
|\V| @>>> [V/S] \simeq \V.
\end{CD}
$$

In addition the gerbe $\G$ on $\Y$ induces a $H-$flux $H$ on 
$|\Y|$ as in Ref.\ \cite{Bunke1}.
The stack T-dual of Bunke and coworkers (Ref.\ \cite{Bunke1})  
is the stack associated to the simplicial bundle which is the 
T-dual of the  simplicial circle bundle $|\Y| \to |\V|$ with $H-$flux $H$ 
on $|\Y|.$

We argue that the Topological T-duality of Mathai and Wu for spaces
with arbitrary $S^1$-actions uses a similar prescription but for the 
the Haefliger-Milnor classifying space.

Let $Y$ be a non-free $S^1$-space. Let $\Y$ be the underlying
stack. Consider the stack $\V = [Y/S^1]$. This has a natural presentation
as the transformation groupoid 
$\VV = [(Y \times S^1) \rightrightarrows Y]$. This stack has a
natural\footnote{See Ref.\ \cite{Noohi3}, Sec.\ (4.3), also the 
proof of Thm.\ (6.3) for the definition and properties of classifying
spaces.} classifying space $B\VV$ (the Haefliger-Milnor 
Classifying Space) associated to this groupoid which is
given by the Borel construction $Y \times_{S^1} ES^1$. 
The principal bundle of stacks $p:\Y \to \V \simeq [Y/S^1]$ gives a principal 
bundle of spaces 
$E\VV \simeq (Y \times ES^1) \to B\VV \simeq (Y \times_{S^1} ES^1)$ 
by pullback and a $2$-cartesian square\footnote{See Ref.\ \cite{Noohi3} 
Sec.\ (4.3) and cartesian square after Lemma (4.1).}
\begin{gather}
\begin{CD}
            E\VV=(Y \times ES^1)  @>f>>     \Y \equiv Y\\
            @VqVV          @VVV \\
            B\VV=(Y \times_{S^1} ES^1)  @>{\phi}>> \V \simeq [Y/S^1]
\end{CD}
\label{GrpSt}
\end{gather}

where the space in each row is a classifying space\footnote{See
Prop.\ (6.1) of Ref.\ \cite{Noohi3}.}.
Also, the map $\phi$ is {\em natural} (See Ref.\ \cite{Noohi3}, before
Sec.\ (4.2)). (In general, the Haefliger-Milnor classifying space of a stack
depends on the groupoid presentation of that stack. 
For a quotient stack
there is a natural choice of the associated groupoid and hence a
natural choice of the associated classifying space.) 

The bundle of classifying spaces obtained in Eq.\ (\ref{GrpSt}) above
is the same as the principal circle bundle in the first column of the
diagram Eq.\ (\ref{MaWuNS}). We had remarked above that the
T-dual of Mathai and Wu could be obtained from this principal
circle bundle by the usual Topological T-duality for circle bundles
(see, for example, Ref.\ \cite{Bunke}).
 
Hence, from the work of Mathai and Wu (Ref.\ \cite{MaWu}) ( when the T-dual
space $\hat{Y}/S^1$ is well-defined), the Topological
T-dual of Mathai and Wu can be obtained from the
the Topological T-dual of the associated principal circle bundle over
the Haefliger-Milnor classifying space above. When $\hat{Y}/S^1$
is {\em not} well defined in the prescription of Mathai and Wu,
the formalism of this section gives a {\em stack} for the T-dual, the stack
associated to the T-dual of the principal bundle of classifying spaces 
$Y \times ES^1 \to Y \times_{S^1} ES^1.$ In this case, the T-dual
obtained here is the {\em stack} $[\hat{Y}/S^1]$ where
$\hat{Y}$ is Mathai and Wu's T-dual.  Mathai and Wu's T-dual
can be naturally recovered as the classifying space $\hat{Y}$ of this T-dual stack
$[\hat{Y}/S^1].$

For either of
these prescriptions of Topologial T-duality, then, for a given stack
$\Y$ with $S^1$-action (which may be the stack associated to a space),
we consider $\Y \to \V \simeq [\Y/S^1].$
The Topological T-dual of $\Y$ is the stack
associated to the T-dual of some natural principal bundle of
classifying spaces $p:Y \to V.$ That is, we find the T-dual
bundle $p^{\#}: Y^{\#} \to V$ of the above bundle $p$ and find a stack
$\Y^{\#}$ such that $Y^{\#} = \Theta(\Y^{\#}).$

It is interesting to ask whether the T-dual stack depends
on this choice of classifying space. We show that
this is not the case. We first need a lemma:
\begin{lemma}
Let $p_i:Y_i \to V_i, i=1,2$ be principal circle bundles over
$V_i.$  Let  $f:V_1 \to V_2$ be a map which is 
also a weak homotopy equivalence.
Suppose that $f$ induces a pullback square:
\begin{equation}
\begin{CD}
Y_1 @>>\tilde{f}>  Y_2 \\
@VVp_1V          @VVp_2V \\
V_1 @>> f> V_2. \label{V1V2Diag} 
\end{CD}
\end{equation}

Then, the set of pairs over $V_1$ and the set of
pairs over $V_2$ are isomorphic by an
isomorphism $B_f$ 
which commutes with Topological T-duality.
\label{LemTTDWHE}
\end{lemma}
\begin{proof}
Let $Y_i \to V_i$ be principal circle bundles
over $V_i$ such that $f$ induces a pullback
square:
$$
\begin{CD}
Y_1 @>>\widetilde{f}>  Y_2 \\
@VVp_1V          @VVp_2V \\
V_1 @>> f> V_2 .
\end{CD}
$$
Let $H_i \in H^3(Y_i, \KZ)$ be the $H$-fluxes on the
$Y_i.$ By definition, $(p_i, H_i)$ are ` pairs' (in the sense of
Ref.\ \cite{Bunke}) over
$V_i, i=1,2.$  Since $f$ is a weak homotopy equivalence, 
$f$ induces isomorphisms  
$f_i:\pi_i(V_1) \to \pi_i(V_2), \forall i \geq 0$ by
definition. This implies that all the cohomology groups of
the $V_i$ are isomorphic.  In particular, every principal
bundle over $V_2$ induces a bundle with the
same characteristic class over $V_1$ by pullback.
Also, every bundle over $V_1$ is the pullback of
a bundle over $V_2$ by $f^{\ast}.$

In addition, the above commutative diagram induces 
a natural  isomorphism of Leray-Serre spectral sequences for the spaces 
$Y_i,$ so that their cohomologies are also isomorphic (see Prop.\ (1.12)
of \cite{SSAT}).

This implies that every $H$-flux $H_2$ on $Y_2 \to V_2,$
pulls back by the induced map on $H^3(Y_2, \KZ)$ to one and only
one $H$-flux on $Y_1 \to V_1.$ Thus there is a natural bijection
$B_f:P(V_2) \to P(V_1).$

Replacing the above square by the T-dual square
$$
\begin{CD}
Y_1^{\#} @>>\widetilde{f}^{\#}> Y_2^{\#} \\
@VVp_1V           @VVp_2V \\
V_1 @>>f> V_2.
\end{CD}
$$
and applying the above argument shows that there is a
natural dual isomorphism which induces isomorphisms
on the total spaces of $Y_i^{\#}.$ 

This isomorphism induces a bijection (also denoted $B$) of dual pairs in
a manner similar to the above. 

Further, we need to show that the bijections
$B_f$ above commute with the action of Topological T-duality
$B_f \circ T \circ P(V_2) = T \circ B_f \circ P(V_2).$
First note that the morphisms $\tilde{f}, \tilde{f}^{\#}$  are induced
from $f$ just by pullback. 

Given the data in the statement of the theorem and in
diagram in Eq.\ (\ref{V1V2Diag}) above,
we can view the pairs on each of the $V_i$'s as a map $h_i:V_i \to R,$
where $R$ is the classifying space of pairs of Bunke et al 
(see Ref.\ \cite{Bunke} above). The above diagram
may be obtained from the pullback of the universal bundle on $R$ along the commutative
diagram associated to the equation $h_1 = h_2 \circ f.$

If one composes the maps $h_i$ with the T-duality transformation $T:R \to R,$
the diagram at the end of the previous paragraph extends to a more
complicated commutative diagram. Examining the pullback of a pair
over $V_2$ along the two obvious routes will show that the image of the
T-dual of a pair is the T-dual of the image of that pair for any pair.

This proves the result.
\end{proof}

\begin{theorem}
Let $p:\Y \to \V$ be a principal circle bundle of stacks. 
Let $\G \to \Y$ be a circle gerbe on $Y.$ 
The T-dual stack exists and doesn't depend on the choice of
the classifying space for $\V.$
\label{ThmTDUniq}
\end{theorem}
\begin{proof}
The proof depends on Lemma (\ref{LemTTDWHE}) above.
In addition, we need the notion of a `pair' from 
Ref.\ \cite{Bunke1} here. Let $p:\Y \to \V$ be a
principal circle bundle of stacks. Let $\G \to \Y$ 
be a gerbe on $\Y$ with band $S^1.$ We consider the
pair $(p,\G).$

By Lemma (\ref{LemTTDWHE}) above,
there is a bijection between pairs over one
classifying space for $\V$ and pairs over another
classifying space for $\V.$ In particular, this implies that 
the T-dual pair is uniquely
specified for any classifying space for $\V$ once it is specified 
for one choice of a classifying space for $\V.$
By the work of Bunke and coworkers in Ref.\ \cite{Bunke1},
Thm. (1.1) and Def. (1.2), the stack T-dual pair exists and
is unique if we use simplicial classifying spaces. 

Note that any two classifying spaces
are weakly homotopy equivalent (see Ref.\ \cite{Noohi3}).
Now, if we change the classifying space, by Lemma\ (\ref{LemTTDWHE})
above, the set of pairs over the two classifying spaces are isomorphic
and the result doesn't
depend on the choice of weak homotopy equivalence.  The above
proof of Bunke and coworkers
can be used to ensure the existence
of T-dual pairs no matter which classifying spaces were used.
\end{proof}

\section{The Gysin Sequence}
\label{secGysin}

The results of Bunke et al \cite{Bunke2} are  for
principal bundles over orbispaces. Here, the base
stack is Deligne-Mumford, i.e., locally a quotient
of a topological stack by the action of a discrete group 
(see Ref.\ \cite{Heinloth}). We wish to study principal
bundles over spaces which do not have this property,
for example, the circle actions on the three-manifolds 
$Y$ in the next section.  Here, the base stack would be
$[Y/S^1]$ which need not be a Deligne-Mumford.
(Note that the method of Mathai and Wu does not have such a restriction.)

In this section, we introduce another method of T-duality, the Gysin
Sequence method, which also gives T-duals agreeing with
either of the two methods above. This method is the extension
of the Gysin sequence method of Topological T-duality for
principal circle bundles to principal circle bundles of stacks 
(see Ref.\ \cite{BEM} and references therein for details).
We will use this method later in the paper. 

Given a pair $([p],H)$ consisting of a principal 
circle bundle $p:E \to W,$ with a $H$-flux 
$H \in H^3(E,\KZ)$ as in Ref.\ \cite{Bunke},
we may compute the characteristic class of
the T-dual using the Gysin Sequence as in
Ref.\ \cite{BHM}, Sec. (2.4).

For the above principal circle bundle, we have
the following sequence (Gysin sequence):
\begin{gather}
 \cdots \to H^2(E,\KZ) \overset{p_!}{\to} H^1(W,\KZ) \overset{ \cup [p]}{\to} H^3(W,\KZ) \overset{p^{\ast}}{\to}
H^3(E,\KZ) \overset{p!}{\to} H^2(W,\KZ)
\overset{\cup [p]}{\to} H^4(W,\KZ) \to \cdots \nonumber
\end{gather}
The above data for a pair defines a class $H$ in $H^3((E,\KZ).$ The
image of this under $p_!$ gives a class $[\hat{p}]$ in $H^2(W,\KZ).$
By exactness, $[\hat{p}] \cup [p] = 0.$ Let $\hat{p}:\hat{E}
\to W$ be the circle bundle over $W$ with characteristic class 
$[\hat{p}].$ Writing the Gysin sequence for $\hat{p},$ since
$[p] \cup [\hat{p}] = [\hat{p}] \cup [p]  = 0, $ by exactness
there exists $\hat{H} \in H^3(\hat{E},\KZ)$ such that $p_{!}(\hat{H}) = [p].$
This is the characteristic class of the $H$-flux $\hat{H}$ of the T-dual bundle. Since
$[p]$ is in the kernel of $\cup [\hat{p}],$ by exactness
 this element is unique up to an element of $\hat{p}(H^3(W,\KZ)).$  

Is it possible to T-dualize principal circle bundles of stacks using
the above method? 
For this to be possible a cohomology Gysin sequence for a principal circle 
bundle of stacks $p:\E \to \F$ is needed.
Ref.\ \cite{Ginot} constructs a homology Gysin sequence from
a principal circle bunde of stacks.  
However, for the above construction we need a cohomology
Gysin sequence which we construct  below.
 
We also show that for stacks $\E$ which are the stacks
associated to a semi-free $S^1$-space $E,$ a `T-dual' stack
may be constructed using a stack analogue of the Gysin sequence
argument above. We show that
that this T-dual stack naturally is the stack associated to the T-dual 
of $E$ which would be obtained in the formalism of Mathai and Wu
(see Ref.\ \cite{MaWu}). We show at the end of this section that
the stack T-dual defined in this section agrees with the T-dual of Bunke 
et al.\ in the examples of Thm.\ (\ref{ThmKK}). This would be expected
due to Thm.\ (\ref{ThmTDUniq}) above.

We now give a few defnitions which will be needed for what follows.
These are from Secs.\ (3,4,6,7,8) of the  paper by Behrend, Ginot, Noohi and Xu
Ref.\ \cite{BGNX}. 

First, note that it is possible to define a vector bundle, an orientation on a vector bundle
and a metric on it for a topological stack (We give the definition from
Ref.\ \cite{BGNX} below.)

\begin{definition}
Let $\X$ be a topological stack. A real {\bf vector bundle} on $\X$
is a representable morphism of stacks $\C \to \X$ which makes $\C$
a vector space object relative to $\X.$ That is, we have an addition
morphism $\C \times_{\X} \C \to \C$ and an $\KR$-action 
$\KR \times \C \to \C,$ both relative to $\X,$ which satisfy the
usual axioms. A complex vector bundle is defined analogously.
\end{definition}

Just as for principal circle bundles, specifying a vector bundle $\C \to \X$
is equivalent to specifying a vector bundle $E_U \simeq U \times_{\X} \C \to U$ 
for every $f:U \to X$ with $U$ a topological space.

Given a vector bundle over a stack we may define an orientation and a metric on
it as in Ref.\ \cite{BGNX}, Secs.\ (3,4).

We now define two properties of morphisms of stacks which are very useful
for the development of the Gysin sequence for stacks. (These
definitions are Def.\ (6.1-6.2) of Ref.\ \cite{BGNX}.)

\begin{definition}
Let $f:\X \to \Y$ be a morphism of topological stacks and 
$\C \to \Y$ a metrizable vector bundle over $\Y.$ A lifting
$i:\X \to \C$ of $f$ is called {\bf bounded} if there is a choice
of metric on $\C$ such that $i$ factors through the unit disk
bundle of $\C.$ 

A morphism $f:\X \to \Y$ of topological
stacks is called {\bf bounded proper} if there exists a metrizable
orientable vector bundle $\C \to \Y$ on $\Y$ and a bounded lifting
$i$ as above such that $i$ is a closed embedding.

A bounder proper morphism $f:\X \to \Y$ is called {\bf strongly 
proper} if every orientable metrizable vector bundle $\G$ on
$X$ is a direct summand of $f^{\ast}(\G')$ for some orientable
metrizable vector bundle $\G'$ on $\Y.$ 
\end{definition}

To a morphism $f:\X \to \Y$ of topological stacks, we associate a 
category $C(f)$ as follows: The objects of $C(f)$ are morphisms
$a:\cL \to X$ such that $fa:\cL \to \Y$ is bounded proper. A morphism
in $C(f)$ between $a:\cL \to \X$ and $b:\W \to \X$ is a homotopy
class (relative to $\X$) of morphisms $g:\cL \to \W$ over $\X.$

It can be proved that $C(f)$ is cofiltered (see Ref.\ \cite{BGNX}, 
Sec.\ (7.1)).
We choose, for each object $a: \cL \to \X,$ a vector bundle
$\C \to \Y$ through which $fa$ factors as in the above
definition of bounded proper morphisms.

The {\bf bivariant singular homology} of an arbitrary
morphism $f: \X \to \Y$ is the $\KZ$-graded abelian
group
$$H^{\bullet}(\X \overset{f}{\to} \Y)
= \varinjlim_{C(f)} H^{\bullet + {\mbox{\tiny rk } \C}}(\C, \C - \K).$$ 
Here, any morphism $\phi:\K=\cL \to {\cL}'$ in $C(f)$
gives as natural graded pushforward homomorphism
$\phi_{\ast}:H^{\bullet + m}(\C, \C - \cL) \to
H^{\bullet + n}(\C', \C' - {\cL}')$ where 
$m = {\mbox{rk}} \C$ and
$n = {\mbox{rk}} \C'$ as shown in Ref.\ \cite{BGNX}, Sec.\ (7.1).

We now come to the definition of an {\bf adequate} morphism of
stacks. This will be repeatedly used below. (This is Def.\ (7.4) of 
Ref.\ \cite{BGNX}.)

\begin{definition}
A morphism $f:\X \to \Y$ of topological stacks is called {\bf adequate}
if in the cofiltered category $C(f),$ the subcategory consisting of
$a:\cL \to X$ such that $fa: \cL \to \Y$ is strongly proper is cofinal.
\label{DefAdeq}
\end{definition}

We refer the reader to Sec.\ (7) of Ref.\ \cite{BGNX} for the
Bivariant theory of topological
stacks including the product of bivariant classes  (defined after Ex.\ (7.5) in
Ref.\ \cite{BGNX}) as these are needed to understand the 
construction of the cohomology Gysin sequence below.
 
We define below normally nonsingular morphisms of topological stacks.
(This is Def.\ (8.15) of Ref.\ \cite{BGNX}.)
\begin{definition}
We say that a representable morphism $f:\X \to \Y$ of
stacks is {\bf normally nonsingular}, ({\bf nns} for short),
if there exist vector bundles $\N$ and $\C$ over the stacks
$\X$ and $\Y,$ respectively, and a commutative diagram
\begin{equation}
\begin{CD}
\N @>>i> \C \\
@AAsA     @VVpV \\
\X @>>f> \Y
\end{CD}
\end{equation}
\end{definition}
where $s$ is the zero section of the vector bundle $\N,$ $i$ is an open
embedding, and $\C$ is oriented.

We refer the reader to Sec.\ (8.3) of Ref.\ \cite{BGNX} for 
the theory of normally nonsingular morphisms of
stacks. 

We define orientation in the sense of the bivariant theory
of stacks of a map of stacks $f: \X \to \Y$ below. The
definitions below are Defs.\ (8.20-21) of Ref.\ \cite{BGNX}.
\begin{definition}
Let $f:\X \to \Y$ be a strongly proper morphism.A bivariant class
$\theta \in H(\X \overset{f}{\to} \Y),$ not necessarily homogenous,
is called a {\bf strong orientation} if for every $g: \Z \to \X,$ 
multiplication by $\theta$ is an isomorphism 
$H(\Z \overset{g \circ f}{\to} \Y) \overset{\sim}{\to}
H(\X \overset{f}{\to} \Y).$

A strongly proper morphism $f: \X \to \Y$ of topological stacks
is called {\bf strongly oriented,} if it is normally nonsingular and it 
is endowed with a strong orientation $\theta_{f} \in H^c(f),$ where
$c = \mbox{ codim }f.$ A topological stack $\X$ is called strongly oriented if
the diagonal $\Delta: \X \to \X \times \X$ is strongly oriented. In this case,
we define $dim \X := codim \Delta.$
\end{definition}

We refer the reader to Defs.\ (8.20-21) onwards of Ref.\ \cite{BGNX} 
for the further development of the above, this will be needed below.

We assume the reader is familiar with
the construction and basic properties of the 
Gysin map associated to a bivariant class in Ref.\ \cite{Ginot}, 
Sec.\ (9.1-9.2) and the construction of $G$-equivariant Gysin maps
in Sec. (9.3) of the same paper.  In addition we assume the development
of the transfer map and the homology Gysin sequence for $S^1$-stacks based on
this in Sec.\ (8) of Ref.\ \cite{BGNX} (We briefly outline some of these ideas
from Refs.\ \cite{Ginot, BGNX} in the proof  below.).

\begin{lemma}
Let $q:\E \to \Y$ be a principal $S^1$-bundle 
of stacks. There is a cohomology Gysin sequence for this bundle if $q$ 
adequate.
\label{LemGysin}
\end{lemma}
\begin{proof}

By Ref.\ \cite{BGNX}, Prop.\ (8.4), every $S^1-$stack 
$\E$ possesses a homology stack Gysin sequence
\begin{gather}
\ldots \to H_i(\Y) \overset{q_{\ast}}{\to} H^{S^1}_i(\Y)
\overset{\cap c}{\to} H^{S^1}_{i-2}(\Y) \overset{T}{\to}
H_{i-1}{\Y} \overset{q_{\ast}}{\to} H^{S^1}_{i-1}(\Y) 
\to \ldots
\label{GysnHmolgy}
\end{gather}
where $T$ is the homology transfer map (also
denoted $T^{S^1}$).
We are looking for a cohomology stack Gysin sequence.
The reason this might not exist is that
it is not clear that a cohomology transfer map $T_{S^1}$ can be
defined for arbitrary $f.$ 

We briefly review the derivation of the cohomology transfer
map from Ref.\ \cite{BGNX} here:
The transfer map $T_{S^1}$ is obtained from
the definition of 
$\theta_{!}$ of Ref.\ \cite{BGNX}, Sec. (9.1),
after Eq.\ (9.1). 

To define $\theta_{!},$
the authors of Ref.\ \cite{BGNX},
consider the $2$-cartesian square
(this is Eq.\ (9.1) of Ref.\ \cite{BGNX})
\begin{gather}
\begin{CD}
{\X}' @>>f'> {\Y'}\\
@VVV        @VVuV \\
\X @>>f> \Y.
\end{CD}
\label{EqTrans1}
\end{gather}

For this diagram of stacks the authors define
the cohomology Gysin homomorphism as 
$\theta_{!}(b) = f'_{\ast}(b \cdot u^{\ast}(\theta)),$
for $b \in H^j(\X') = H^j(\X' \overset{ \mbox{id} }{\to} \X').$
However, this requires that the map $f'$ is adequate or the
above product need not be defined (by Ref.\ \cite{BGNX},
after Ex.\ (7.5)). 

The cohomology Gysin map $T_{S^1}$ is obtained by applying the
above to the $2$-cartesian square
\begin{equation}
\begin{CD}
\E @>>q> [S^1 \backslash {\E}]\\
@VVV        @VVuV\\
{\ast} @>>q> [S^1 \backslash {\ast}].
\end{CD}
\label{GysinCD}
\end{equation} 
The transfer map in cohomology $T_{S^1}$ is $\theta_{!}$ for
the diagram in Eq.\ (\ref{EqTrans1}) calculated for the diagram in Eq.\ (\ref{GysinCD}).
Hence, by the above, for this map to be well-defined we require
that $q$ be adequate. Now, by Lemma\ (\ref{StLem1}) above,
considering the principal circle bundle
$q: \E \to [S^1 \backslash \E]$ is exactly equivalent to considering
the principal circle bundle $q:\E \to \Y$ since by the
results there $\Y \simeq [S^1 \backslash \E]$ always.
Hence, it is enough to assume that $q:\E \to \Y$ is
adequate.


The proof of Prop.\ (8.4) of Ref.\ \cite{Ginot} may now be followed
with $H_{\ast}$ replaced by $H^{\ast}$. 
It is clear that the full proof goes through.
Let $Z \to [S^1 \backslash \E]$ be a classifying space for 
$[S^1 \backslash \E]$ and $Y \to \E$ be the classifying space
for $\E$ obtained by pullback along $q$.
Let $c$ be the Euler class of disk bundle associated to 
the principal bundle $Y \to Z$, then
we obtain the following Gysin Sequence:
\begin{equation}
\ldots \to H^{i-1}_{S^1}(\E) \overset{q^{\ast}}{\to} H^{i-1}(\E)
\overset{T_{S^1}}{\to} H^{i-2}_{S^1}(\E) \overset{\cup c}{\to} H^{i}_{S^1}(\E)
\overset{q^{\ast}} \to H^i(\E) \to \ldots
\label{CohoGysin}
\end{equation}
from the cohomology Gysin sequence associated to $Y \to Z$ 
under the identifications 
$H^i(Z) \simeq H^i([S^1 \backslash \E]) \simeq H^i_{S^1}(\E)$,
and $H^i(Y) \simeq H^i(\E)$ exactly as in Ref.\ \cite{Ginot}, Prop.\ (8.4).
\end{proof}

We now prove that  under some conditions the quotient map $p$ associated
to a principal bundle of stacks $p:\E \to [S^1 \backslash \E]$ is adequate.

We need some theorems from Ref.\ \cite{BGNX}.
Recall from Sec.\ (\ref{secTDKK}) above that the semi-free spaces
we consider are all total spaces of $KK$-monopoles. In particular,
they are all oriented orbifolds.

They are all strongly oriented by the following proposition:
(This is Prop.\ (8.35) of Ref.\ \cite{BGNX} )

\begin{theorem}
Let $\X$ be a paracompact orbifold whose tangent bundle
is oriented. Then the diagonal $\X \to \X \times \X$ is
strongly oriented and, in particular, $\X$ is naturally
oriented. 
\label{BGNX835}
\end{theorem}

In addition we need the following Proposition (Prop. (8.32) from Ref.\ \cite{BGNX} ).
Let $G$ be a compact Lie group and $X$ and $Y$ smooth $G$-manifolds,
with $\X = [X/G]$ and $\Y=[Y/G]$ the corresponding quotient stacks. 

\begin{theorem}
Let $X,Y$ be as above. Assume further that $X$ and $Y$ are oriented and
that the $G$-actions are orientation preserving. Then, every normally nonsingular
diagram for $f:\X \to \Y$ is naturally oriented. In particular, when $f$ is strongly proper,
we have a strong orientation class $\theta_f \in H^c(f), c = \mbox{ dim } Y - \mbox{ dim }X.$
Furthermore, this class is independent of the choice of the normally nonsingular diagram.
\label{BGNX832}
\end{theorem}

We may now prove the above:
\begin{lemma}
Let $\E$ be a stack associated to a semi-free space. 
Let $\Y = [S^1 \backslash \E]$ and let $p:\E \to \Y$ be the quotient map. 
Let $\V$ be the associated vector bundle to the principal bundle
$p:\E \to \Y.$ Suppose $\V$ is metrizable. Further, 
suppose $\Y$ is strongly oriented. Then $p$ is strongly oriented. 
Also, $p$ is adequate.
\label{LemAdeq}
\end{lemma}
\begin{proof}
Firstly, $\E$ is strongly oriented due to Prop.\ (\ref{BGNX835}) above.
Secondly, $\Y$ is strongly oriented by assumption.
Thirdly,  we will argue that $p$ is strongly proper and normally nonsingular. 
Hence, $p:\E \to \Y$ has a strong orientation class by Prop.\ (\ref{BGNX832}).

By Lemma (\ref{StLem1}) above, $p:\E \to \Y$ is a principal $S^1$-bundle. 
By Ref.\ \cite{Ginot},
it is also representable (By Prop.\ (4.7) of Ref.\ \cite{Ginot}). 
Let $q:\V \to \Y$ be the associated vector bundle. 
$\V$ is metrizable by assumption.
Let $s:\E \to \V$ be the embedding of $\E$ as the
unit sphere bundle in $\V.$ Then, the following diagram commutes
\begin{equation}
\begin{CD}
            \V @>{id}>>     \V \\
            @AA{s}A          @VV{q}V  \\
            \E  @>{p}>> \Y.
\end{CD}
\end{equation}
This is the required normally nonsingular diagram for $p.$

It is clear that $p$ is bounded proper. 
It remains to prove that $p$ is strongly proper.
Let $w:\C \to \E$ be an orientable metrizable vector bundle on $\E.$ 
Choose a classifying space $Y \to \Y$ and pull back $Y$ to $\E$ along $p$ to 
obtain a classifying space $E \to \E.$ This gives a 
principal bundle $E \to Y.$ Pulling $\V$ back to $Y$ we obtain a bundle 
$\V \underset{\Y}{\times} Y \to Y.$ The pullback of this bundle
along the map $E \to Y$ is trivial since the bundle has the same Euler 
class as the bundle of stacks $\V \to \Y$ 
(this follows from the definition of the Euler class, see Behrend et al.\ 
Ref.\ \cite{BGNX} Ex.\ (8.26)).
It follows that $p^{\ast}(\V)$ is a trivial bundle. Hence, taking charts,
there is an integer $n > 0$ such that 
$\C$ is a direct summand of $(p^{\ast}(\V))^n \simeq p^{\ast}(\V^n).$ 
Now, $\V$ is the vector bundle associated to $\E,$ and, since $p$ is
orientable, so is $\V.$
Further, $\V$ is metrizable by assumption and hence, so is $\V^n.$
Hence, $p$ is strongly proper. By the above $p$ is normally nonsingular.
Hence, $p$ has a strong orientation class by Prop.\ (\ref{BGNX832}) above. 
By Ex.\ (7.5) (1) of Ref.\ \cite{BGNX}, $p$ is adequate.
\end{proof}

We had noted above that the Gysin sequence can be used to calculate
the characteristic class and $H-$flux of the T-dual principal bundle
of stacks. It is interesting to note that the T-dual data doesn't depend on
the choice of classifying space made in this calculation.

Thus we may use the above Gysin sequence to obtain a T-dual by analogy
with the usual Gysin sequence argument for Topological T-duality:
Given a principal bundle of stacks $p: \E \to \Y$ 
with Euler class $[p] \in H^2(\Y)$ and a 
gerbe $\G \to \E$ with characteristic class $[H] \in H^3(\E)$ 
(see Ref.\ \cite{Heinloth} Sec. (5), Prop. (5.8)), 
we can define the T-dual to be the principal bundle
of stacks $p^{\#}: \E^{\#} \to \Y$
whose Euler class is $T_{S^1}([H])$ by the Gysin sequence method above.
The T-dual $H$-flux would then be a gerbe on $\E^{\#}$ with characteristic 
class $[H^{\#}]$ such that $T_{S^1}([H^{\#}]) = [p].$

To show that this T-dual exists and is unique, we use Thm.\ (\ref{ThmTDUniq})
above.We may pick a classifying space $Y \to \Y.$ By definition, 
the class $T_{S^1}([H]) \in H^2(\Y)$ determines a class in
$H^2(Y,\KZ)$ since the two cohomology groups are naturally
isomorphic. This class in turn determines, up to an isomorphism
of pairs, a dual principal circle bundle of spaces 
$E^{\#} \to Y$ with dual $H-$flux $H^{\#}.$
By the work of Bunke and coworkers (see Ref.\ \cite{Bunke1}),
this  uniquely determines 
up to isomorphism of pairs in the sense of Ref.\ \cite{Bunke1},
the T-dual principal circle bundle of
{\em stacks} $p^{\#}:\E^{\#} \to \Y$ and T-dual gerbe
$\G \to \E^{\#}$ with characteristic class 
$H^{\#} \in H^2(E^{\#},S^1).$  By Thm.\ (\ref{ThmTDUniq}),
this T-dual doesn't depend on the choice
of the classifying space.

It is interesting to note that
the Gysin sequence method above can be
used to obtain the T-dual in both the method of Mathai and Wu 
(see Ref.\ \cite{MaWu}) and Bunke and coworkers
(see Ref.\ \cite{Bunke1}) just by changing the choice of
classifying space in the calculation.

We would like to apply the above to some concrete examples.
We apply these to the $KK$-monopole spaces in Sec.\ (\ref{secTDKK}) 
above. For all the examples of $KK$-monopoles in 
Thm.\ (\ref{ThmKK}) above, the total space is the stack 
$[CS^3/\KZ_k]$ associated to an oriented orbifold and 
hence strongly oriented by Prop.\ (8.35) of Ref.\ \cite{BGNX}.
 From Thm.\ (\ref{ThmKK}) above, the quotient stack by the
$S^1$-action is always $[CS^3/S^1]$
and any vector bundle $\V$ over $[CS^3/S^1]$
is metrizable by Ex.\ (3.3) of Ref.\ \cite{BGNX}. 
Also $[CS^3/S^1]$ is strongly oriented by Prop.\ (8.33) 
and Def.\ (8.21) of Ref.\ \cite{BGNX}. Hence, in all the 
examples of $KK$-monopoles calculated above, the bundle maps
$p_k:[CS^3/\KZ_k] \to [CS^3/S^1]$ and $p:\underline{CS^3} \to [CS^3/S^1]$  
are adequate. Thus, we may use the Gysin sequence argument above
to obtain the T-dual.

\begin{corollary}
For every $k > 1$, $k$ a natural number,
the T-dual of the principal bundle of stacks 
$p_k:[CS^3/\KZ_k] \to [CS^3/S^1]$ 
using the Gysin Sequence method above is the principal bundle of stacks
$q:([CS^3/S^1]\times S^1) \to [CS^3/S^1]$ with $H$-flux.
\label{CorKKNew}
\end{corollary}
\begin{proof}
The Borel construction $CS^3 \times_{S^1} ES^1$ is a classifying space for
the stack $[CS^3/S^1]$. Also, the Borel construction
gives $(CS^3/S^1) \times B\KZ_k$ as the classifying space for the stack
$[CS^3/\KZ_k]$. By definition, the principal bundle of classifying
spaces associated to the principal bundle of stacks $p_k$ above is the
principal circle bundle $(CS^3/S^1) \times B\KZ_k \to (CS^3/S^1) \times BS^1$. 
Using the Gysin Sequence for this principal bundle of spaces, since there is
no $H-$flux on $(CS^/S^1) \times B\KZ_k,$ 
the T-dual would directly be the principal circle bundle 
$q:((CS^3/S^1) \times BS^1) \times S^1 \to ((CS^3/S^1) \times BS^1)$
with $H$-flux. Now, by the above,
$((CS^3/S^1) \times BS^1)$ is the classifying space associated
to the stack $[CS^3/S^1]$. From the proof of Thm.\ (\ref{ThmKK}) above,
the principal circle bundle $q$ above is the bundle of classifying spaces
associated to the principal bundle of stacks $q:([CS^3/S^1] \times S^1)
\to [CS^3/S^1]$. The fact that there is a $H$-flux present on the bundle
of classifying spaces implies that there is a gerbe on this T-dual bundle of stacks. 
\end{proof}


It is also  interesting to note that if one could find another natural classifying
space apart from the simplicial and Haefliger-Milnor classifying spaces
above (perhaps for a restricted class of stacks),
one could use the above Gysin sequence argument to define a procedure to calculate
the T-dual entirely using those classifying spaces. In particular, due to the above 
we could always argue that the stack theory T-dual calculated using 
principal circle bundles over the new classifying spaces would always
agree with the stack theory T-duals calculated using the simplicial
classifying space. This might be more
convenient in some specialized applications.

\section{The T-dual of a three-manifold with $S^1$-action}
\label{sec3Mfd}
In Sec.\ (\ref{secTDKK}) above, we had studied the T-dual of
the space $CS^3$ with its natural semi-free circle action.
In that section we had noted that the topological space
associated to the manifold which was the spacetime 
background was homeomorphic to $CS^3 \times \KR.$ 
Note that $CS^3$ is also the underlying topological space
of a three-manifold with circle action.
It is interesting to ask what would
happen if $CS^3$ were replaced by an arbitrary compact
three-manifold since circle actions on compact three-manifolds
are completely classified (see Ref.\ \cite{Raymond, Fintushel}).

In this section we calculate T-duals of 
arbitrary compact three-manifolds with an arbitrary circle action. 
These may be turned
into four-manifolds by adding a $\KR$-factor
as for $CS^3 $ above. (For example, one could view the result as a
warped product of $\KR$ with a spatial slice).
However, it is necessary to check whether the resulting
space can be the underlying topological space of a
supersymmetric Type II string theory background.
Whether these are valid string theoretic backgrounds or not,
it is still possible to calculate the Toplogical T-dual of these
spaces. It has been observed that when
the space being T-dualized is the underlying
topological space of a valid supersymmetric Type II background,
the Topological T-dual is homeomorphic to the underlying topological
space of the physical T-dual spacetime background.

In addition to the above, one may view three-manifolds with a
circle action as the event horizons of five-dimensional axisymmetric black hole 
spacetimes with a $U(1) \times U(1)$ symmetry. We discuss this
matter at the end of this section.

In addition, at the end of this section we argue that there is an
isomorphism in the twisted $K$-theory and twisted cohomology of
a given three-manifold with circle action and its T-dual. The argument
is slightly difficult since the base is not an orbifold.

Unlike all the previous examples, the spaces studied in this section
are not semi-free spaces, however,
there is no problem in handling these with either of the above
formalisms: Consider an arbitrary $3$-manifold $W$ with a smooth
circle action without $H$-flux. By Ref.\ \cite{Ginot}, Sec.\ (11),
the associated stack 
$\underline{W}$ to $W$ is a principal $S^1-$bundle of 
stacks $\underline{W} \to [W/S^1]$
over the quotient stack\footnote{See Prop.\  (11.4) of
Ref.\ \cite{Ginot}} $[W/S^1]$ (compare with 
Lemma (\ref{StLem1}) in Sec.(\ref{secTDKK}) above).

We show that it is possible to T-dualize this bundle using the methods
outlined above for any $H$-flux on the $3$-manifold $W.$ 
We begin with an elementary observation:

There is a classification of smooth $S^1$-actions
on three-manifolds (see Ref.\ \cite{Raymond}) which states that any three-manifold
with $S^1$-action is an equivariant
connected sum  of 'simpler' three-manifolds
with possibly non-free $S^1$-actions.  

In the following we need the notion of a homotopy colimit of two spaces
$X, Y.$ (We use Ref. \cite{Dugger} Ex.\ (2.2) in the following). 
Suppose $X,Y$ are two spaces and $A$ is a subset of both
$X,Y.$ We have a diagram  
$D:X \overset{f}{\leftarrow} A \overset{g}{\to} Y.$ 
A {\bf homotopy colimit} of $D$ is a 
pushout of this diagram in the category of spaces:
$$
\begin{CD}
A @>>g>  Y \\
@VVfV       @VVV \\
X @>>> \mbox{hocolim(D)}
\end{CD}
$$

The space hocolim$(D)$ is actually the adjunction space
$X \coprod_f (A \times I) \coprod_g Y$ where $f(a) \sim (a,0)$
and $(a,1) \sim g(a)$ for  $a \in A.$

If $f$ is the inclusion and $g$ is any map, then $\mbox{hocolim}(D)$
is the space $X$ glued to the mapping cylinder of the map $g:A \subset X \to Y.$ 
Let $\mbox{colim}(D)$ be the colimit of the diagram $D.$
It can be shown, (see Ref.\ \cite{Dugger}, Sec.\ (2)) that there is a natural map 
$\mbox{ hocolim}(D) \to \mbox{colim }(D)$ obtained
by collapsing the mapping cylinder $M_f.$
It can also be shown (see Ref.\ \cite{Dugger}, Prop. (13.10) ) that 
$\mbox{hocolim}(D)$ is weakly homotopy equivalent to
$\mbox{colim}(D) \simeq X \cup_g Y.$

We can also define an equivariant version of this construction and then
$\mbox{colim}(D)$ is the equivariant gluing of two spaces $X,Y$ along a map 
$f:A \to A$ where $A \subset X, A \subset Y$ is an invariant set
(see Ref.\ \cite{Bredon}).
 
We would like to determine the Topological T-dual of an equivariant connected
sum of two spaces. We first study the behaviour of 
the equivariant join of two spaces under Topological T-duality.

\begin{theorem}
\leavevmode
\begin{enumerate}
\item Let $X$ be a space with a $S^1$-action. Consider the principal
bundle of stacks $p:\underline{X} \to [X/S^1].$  Let $\phi:V \to [X/S^1]$
be a classifying space for $[X/S^1]$ and $Y \to \underline{X}$ the
classifying space for $X$ obtained by pulling back $\phi$ along
$p.$ Suppose $V  \simeq V_A \cup V_B,$ and consider the
(2)-commutative diagram
$$
\begin{CD}
Y @>>> \underline{X} \\
@VVV          @VVpV \\
V \simeq V_A \cup V_B @>>\phi> [X/S^1].
\end{CD}
$$

Then, we have that $Y \simeq Y_A \cup Y_B$ with $Y_i$ the pullback
of $V_i$ along $p|_{V_i}$  for $i=A,B$ respectively. 
\item Let $X$ and $Y$ be spaces with a $S^1$-action. Let $A \subset X$ be
an invariant set (i.e. $z.A \subseteq A, \forall z \in S^1$). Let $A$ be identified via the identity map with
an invariant subset $A \subset Y.$ \label{itmGlue}
Then, the stack T-dual of the equivariant attaching of $X \cup_A Y$ of $X$ to $Y$ along $A$ with
$H-$flux $H$ is the equivariant attaching of the stack T-dual of $X$ to the stack T-dual of
$Y$ along the T-dual of $A$ with dual $H-$flux $H^{\#}.$ 
\item Let $A$ be the total space of a principal circle bundle. Let $f:A \to A$ be
a map equivariant under the $S^1-$action on $A.$ Let the Topological T-dual of
$A$ be $A^{\#}.$ Let $M_f$ be the equivariant mapping cylinder of $f.$ Then,
the Topological T-dual of $M_f$ with $H-$flux $H$ is $M_f^{\#}$ with $H-$flux
$H^{\#}$ where $f^{\#}:A^{\#} \to A^{\#}$ is an equivariant map. \label{itmMf}
\item Let $X$ and $Y$ be as in Item (\ref{itmGlue}) above. Let $f:A \to A$ be an equivariant
map. Then, the T-dual of $X \cup_f Y,$ 
the equivariant gluing of $X$ to $Y$ along $A$ using $f,$ with $H-$flux $H$
is $X^{\#} \cup_{f^{\#}} Y^{\#}$ with T-dual $H-$flux $H^{\#}$ 
for some map $f^{\#}$ which can be computed.  
\end{enumerate}
\label{ThmEAttach}
\end{theorem}
\begin{proof}
\leavevmode
\begin{enumerate}
\item 
Let $\V \simeq [X/S^1].$
By definition 
$$
Y \simeq V \times_{\V} X \simeq (V_A \cup V_B) \times_{\V} X.
$$
We would like to prove that this is equal to 
$(V_A \times_{\V} X) \cup (V_B \times_{V} X).$
This trivially follows from the gluing property of the fiber
product of stacks (see Ref.\ \cite{Heinloth} for example).

By definition of the fibered product of stacks (see Ref.\ \cite{Heinloth} for example),
for any topological space $U$ we have that:
\begin{gather}
(V_A \cup V_B)  \times_{\V} X(U) \simeq  <(f,f',\psi)| f:U \to (V_A \cup V_B), f':U \to X 
\mbox{ s. t. }\psi: p \circ f' \Rightarrow \phi \circ f> \nonumber \\
 \subseteq (V_A \times_{\V} X)(U) \cup (V_B \times_{\V} X)(U) \nonumber
\end{gather}
 where the second line follows by composing $f$ in the first line with projections to the $V_i,i=A,B$ and
using the definition of the fibered products $(V_i \times_{V} X), i= A,B.$

Also, we have that maps $f:U \to U_A$ and $f':U \to X$ together with $2$-morphisms
$\psi: p \circ f' \Rightarrow \phi_A \circ f$ give maps $U \to U_A \cup U_B$ by composing
$f$ and $f'$ with the obvious inclusions. These maps satisfy the condition
$p \circ f' \Rightarrow \phi \circ f$ since $\phi_i = \phi|_{V_i},i=A,B.$

Hence, the reverse inclusion is true as well and the result follows.

\item 
This is actually an argument about the classifying space
associated to a stack. As was discussed above, the
Toplogical T-dual is computed using the classifying space 
associated with a given stack:
To a principal bundle of stacks we associate a
bundle in which each space is the associated
classifying space. By the previous part of this 
theorem the bundle of classifying spaces associated 
to $X \cup_A Y$ is obtained by equivariantly
attaching the classifying space bundles associated
to $X$ and $Y$ along $A.$ 

Also note that, by definition, $A \subseteq X$ is a union of
(possibly non-free) circle orbits.
Thus, its T-dual $A^{\#}$ is well-defined.
Thus, we have an inclusion of substacks which
are the total spaces of principal bundles of stacks
$$
\begin{CD}
\underline{A} @>>i> \underline{X} \\
@VVp_AV  @VVp_XV \\
[A/S^1] @>>i> [X/S^1].
\end{CD}
$$

Passing to classifying spaces we obtain a
inclusion of spaces which are total 
spaces of principal bundles of classifying
spaces of stacks 
$$
\begin{CD}
Y_A @>>i> Y_X \\
@VVp_AV   @VVp_XV \\
V_A @>>i> V_X.
\end{CD}
$$
Also, we have that $p_A = p_X|_{V_A}.$
Since this is a bundle of {\em spaces},
semi-locality of Topological T-duality for
spaces (as shown in Sec.\ (\ref{secTDKK}),
Thm.\ (\ref{ThmKK}) above) proves that 
the T-dual $A^{\#}$ of $A$  is a substack of
the T-dual $X^{\#}$ of $X.$ By interchanging
the roles of $X$ and $Y$ we see that the same
is true of $A^{\#} \subseteq Y^{\#}.$

This implies that the T-dual of $X\cup_A Y$
is $X^{\#} \cup_{A^{\#}} Y^{\#}.$

%
%

\item By definition, $M_f =  (A \times I)$ with $A \times \{ 1 \}$ 
glued to $A$ via the map $f.$ If the circle action on $A$ only has
free orbits,  it is clear that the circle action on $M_f$ must have
only free orbits. 

Let $A/S^1 = B.$ Then, $M_f$ is the total space of a principal circle bundle over
a base $V.$ $V$ is the mapping cylinder of the map
$f_B: B \to B$ induced by the equivariant map $f:A \to A.$
Hence, the T-dual of $M_f,$ say $Y,$ must be a principal circle bundle
$\pi: Y \to V$ since $M_f$ is one.

Let $H$ be a $H-$flux on $M_f$ and let $U \subseteq I.$
We are T-dualizing along the circle orbits in $A \hookrightarrow M_f.$
By the semi-locality of Topological
T-duality, the T-dual of $A \times U \subseteq M_f$ with
$H-$flux the restriction of $H$ to $A \times U$
is $A^{\#} \times U$ with some T-dual $H-$flux
which is the restriction of $H^{\#}$ to $A^{\#} \times U.$
Hence just by semi-locality of Topological T-duality,
the T-dual of $M_f$  is $A^{\#} \times [0,1)$ glued to 
$A^{\#} \times \{1 \}$ by some map $f^{\#}.$ 
Hence there is a natural surjective map $i_1: A^{\#} \times I \to Y$
given by the quotient map. Also we have the natural inclusion
$i_2:A^{\#} \to Y$ which sends $a$ to the equivalence class
of $(a,1)$ in the above quotient.

It is clear that for every $t \in [0,1),$ 
$f^{\#}(x,t) = (x,1)$ where $(x,t) \in A^{\#} \times [0,1) \subseteq Y.$ 
Note that $f^{\#}$  is  $S^1$-equivariant by definition. 

We now use the universal property of a mapping cylinder to prove
that $Y$ is the mapping cylinder of $f^{\#}.$

Given a map $f^{\#}:A^{\#} \to A^{\#},$ the mapping cylinder $M_{f^{\#}}$
is a pushout
$$
\begin{CD}
A^{\#} @>>f^{\#}> A^{\#}\\
@VVV      @VVV\\
A^{\#} \times I @>>> M_{f^{\#}}.
\end{CD}
$$

The mapping cylinder has the universal property that for
any space $Z,$ and a mapping $g_1: A^{\#} \times I \to Z,$
$g_2:A^{\#} \to Z,$ such that $g_1(x,1) = g_2(f^{\#}(x))$ for
all $x \in A^{\#},$ there is a unique $k:M_{f^{\#}} \to Z$ such
that the composition $A \times I \to M_{f^{\#}} \to Z$ equals
$g_1$ and the composition $A \to M_{f^{\#}} \to Z$ equals $g_2.$
 
Suppose $Z$ is any space and we are given $g_1,g_2$
as above. We note that by the above, we already have
inclusions $i_1:A^{\#} \times I \to Y$ and $i_2:A^{\#} \to Y.$
Also $i_1(a,1) = i_2(f^{\#}(a))$ by definition of $f^{\#}$ above.

Given the above, define $k:Y \to Z$ by $k((a,t)) = g_1(a,t)$ for all
$t \in [0,1].$ It is clear that $k$ is well defined. By uniqueness
of $M_{f^{\#}},$ this implies that $Y  \simeq M_{f^{\#}}$ 
is a mapping cylinder.

\item Let $X$ and $Y$ be $S^1$-spaces and $A$ an invariant
subset of $X$ and $Y.$ Let $f:A \to A$ be an equivariant map.

Consider the given pair $(X \cup_f Y, H).$ We try to
compute its T-dual by replacing the space $X \cup_f Y$ by
a weakly homotopy equivalent space.

Let $D$ be the diagram of $S^1$-spaces and equivariant
maps $X \leftarrow A \overset{f}{\to} Y.$ We replace $X \cup_f Y$ by 
equivariant homotopy colimit of $D,$   $N = \mbox{hocolim}{(D)}.$ 
$N$ is weakly equivariantly homotopy equivalent
to $X \cup_f Y.$ (See discussion before this theorem) by
a natural map $h: N \to X \cup_f Y.$
Hence, by Lemma\ (\ref{LemTTDWHE}) above, the given pair 
$(X \cup_f Y, H)$ induces a unique pair $(N,H').$

Note that the space $N$ contains $X,Y$ and $M_f$ as subspaces. 
In addition, $M_ f$ is the total
space of a principal circle bundle since $f$ is equivariant. Since $N$
is made by equivariantly gluing $X, Y$ and $M_f,$ 
$N$ is a principal bundle over a base space $B.$

The T-dual of the pair $(N,H)$  is another pair $(N^{\#},H^{\#}).$
By semi-locality of Topological T-duality, the space $N^{\#}$ is also a principal circle
bundle over $B,$ since it contains the T-duals
$X^{\#}, Y^{\#}$ and $(M_f)^{\#}$ of 
$X,Y, M_f$ respectively. By Item (\ref{itmMf}) above there is
a map $f^{\#}:A^{\#} \to A^{\#}$ such that $(M_f)^{\#} \simeq M_{f^{\#}}$

Let $D^{\#}$ be the diagram obtained from $D$ by T-dualizing
each space in the diagram. By the semi-locality of
Topological T-duality, the inclusion $X \leftarrow A$ is replaced
by the inclusion $X^{\#} \leftarrow A^{\#}.$ In addition,
the arrow $f:A \to Y$ is replaced by
the arrow $f^{\#}:A^{\#} \to Y^{\#}.$

Thus, the diagram $D^{\#}$ would be 
$X^{\#} \leftarrow A^{\#} \overset{f^{\#}}{\rightarrow} Y^{\#}.$
By the discussion above and also by 
Item (\ref{itmMf}) above, $N^{\#}$  is homeomorphic to
the T-dual homotopy colimit $\mbox{hocolim}(D^{\#}).$

By the discussion before this theorem, this space $N^{\#}$ is weakly homotopy equivalent to
$X^{\#} \cup_{f^{\#}} Y^{\#}$ by a natural map $h^{\#}:N^{\#} \to X^{\#} \cup_{f^{\#}} Y^{\#}.$
 Hence, the pairs on them are in bijection and
the pair $(N^{\#},H^{\#})$ corresponds to a unique
pair $(X^{\#} \cup_{f^{\#}} Y^{\#}, H_1^{\#}).$ This is the T-dual
of $(X \cup_f Y, H)$ by the second part of
Lemma \ (\ref{LemTTDWHE}) above.
\end{enumerate}
\end{proof}

\begin{corollary}
\leavevmode
\begin{enumerate}
\item Let $W_1,W_2$  be two $3$- manifolds with a 
smooth $S^1$-action. Let $W_1 \# W_2$ denote
the equivariant connected sum of $W_1$ and
$W_2$. Then,  the Topological T-dual of $W_1 \# W_2$  with arbitrary $H-$flux 
$H$ is the equivariant connected sum of the T-duals
of $W_1$ and $W_2$ with a $H$-flux on the $W_i$ which is the restriction of
$H$ to $W_i, i=1,2.$
\item Let $X$ be a $CW$-complex with circle action, and let $C_0X$ be the equivariant
cone on $X$ with vertex $x_0.$ Then the topological T-dual
of $C_0X$ is the stack which is the
T-dual of $X$ glued to the stack $\{x_0\} \times S^1$ with a source
of $H$-flux on $\{ x_0 \} \times S^1$.
\end{enumerate}
\label{CorEquivSum}
\end{corollary}
\begin{proof}
\leavevmode
\begin{enumerate}
\item This follows from the Pt.\ (2) of Thm.\ (\ref{ThmEAttach}) above
as the equivariant connected
sum of two spaces is a special case of the equivariant attaching
of two spaces.
\item This follows from Pt.\ (2) of Thm.\ (\ref{ThmEAttach}) above as $C_0X$ is equivariantly
homeomorphic to $(X \times [0,1]) \cup_{\{1\} \times X } x_0.$ 
\end{enumerate}
\end{proof}

Note that in Thm.\ (\ref{ThmEAttach}) and Cor.\ (\ref{CorEquivSum}) above, 
{\em we could have replaced some of the spaces with orbispaces} as in the work of
Bunke and coworkers (see Ref.\ \cite{Bunke1}) and used the simplicial classifying
space for orbispaces. {\em This would not change any of the above results except that
the final answer might not be a space but might be a stack.} 
Thus the above results trivially extend to T-dualizing the gluing of spaces and manifolds 
with circle actions to orbispaces with circle actions. In what follows
we will glue orbispaces and spaces as freely as needed.

Note that in the second part of the Corollary above, if we take $X$ be
a three-manifold, $C_0X$ will be a three-dimensional $CW$-complex
with $S^1-$action. In particular, taking $X=S^3, S^3/\KZ_k$ will recover the
results of Sec.\ (\ref{secTDKK}) above.
 
By Refs.\ \cite{Raymond, Fintushel, OrRay},  any compact three-manifold $M$ with
$S^1$-action has the following 
types of circle orbits only: Free orbits, isolated fixed points, fixed point fibers, 
exceptional orbits and special  exceptional orbits. (The last four orbit types have a stabilizer).

We review each of these below:
\begin{itemize}
\item Free orbits are orbits which have the identity as a stabilizer.
It can be proved (see Ref.\ \cite{Raymond}) that the quotient
map $q: M \to M/S^1$ restricted to the image of the free orbits is actually
a principal circle bundle.

(It is strange (see Ref.\ \cite{Fintushel}) that this classification of
orbit types is also true for three-manifolds which only locally have a
$S^1$-action, i.e. for $O(2)$-spaces.)
 
\item Fixed points have a neighbourhood 
homeomorphic to the 3-disk with an orthogonal
$S^1$-action. 
Fixed point fibers have neighbourhoods
isomorphic to a solid torus $\KD^2 \times S^1$ with
the $S^1$ action:
\begin{gather}
z \cdot (\rho e^{i\theta}, e^{i\psi}) = (z \rho e^{i \theta}, e^{i\psi}), z \in S^1, \theta \in [0, 2 \pi] .
\label{FPNbd}
\end{gather}
Here the set $\{ (0,e^{i \psi}) \}$ is a circle of fixed points.

\item Exceptional orbits have a neghbourhood equivariantly homeomorphic
to $\KD^2 \times S^1$ with the action:
\begin{gather}
z \cdot (\rho e^{i\theta}, e^{i\psi}) = (z^{\nu} \rho e^{i\theta}, z^{\mu} e^{i\psi}), z \in S^1.
\label{ExNbd}
\end{gather}
where $\mu$ and $\nu$ are relatively prime and $0 < \nu < \mu.$
The exceptional orbit itself is the set $\{ 0 \} \times S^1.$
\item Special exceptional orbits have a neighbourhood equivariantly
homeomorphic to $\KD^2 \times_{\KZ_2} S^1$ 
(see Ref.\ \cite{Bredon} for a definition) 
where $\KZ_2$ acts on $\KD^2$ by reflection. 
\end{itemize}

The three-manifold $M$ above can be obtained by
equivariantly attaching (see Ref.\ \cite{Raymond})
the above spaces to a principal circle bundle over a 
two-manifold. 

More precisely, given a compact three-manifold $M$ with $S^1$-action,
if there are $h$ fixed point orbits and $t$ 
special exceptional orbits, we can construct 
$M$ by taking an equivariant connected sum of a
principal circle bundle with the above $S^1$-spaces 
(see Ref.\  \cite{Raymond}) as folows:

Consider the fibration $q:M \to M/S^1.$ Pick open nbds
of the image of the fixed point sets in $M/S^1.$ Pick
an open set W in $M/S^1$ which does not intersect
the images of the fixed sets in $M/S^1.$
Restricted to $W,$ $M$ is a principal circle bundle.

If there are no exceptional orbits, the compact three-manifold
$M$ is isomorphic to $(\Sigma_{g,h,t} \times S^1)/\simeq,$
where $\Sigma_{g,h,t}$ is a compact oriented surface of
index $g,$ with $h+t$ boundary components, $h \geq 1.$+

The equivalence relation $\simeq$ is as follows:
For each of the $h$ boundary circles, the (trivial) circle
bundle over it is collapsed onto the base circle via the
projection map. This corresponds to equivariantly 
attaching\footnote{See Ref.\ \cite{Bredon}} the principal bundle $M|_W$ above to 
a torus neigbhourhood of the fixed point orbits with the circle action in
Eq.\ (\ref{FPNbd}).
The $h$ boundary circles are the fixed
point orbits.

Over each of the remaining $t$ circles,
the fibers of the corresponding (trivial) circle bundle
are quotiented out by the antipodal action.
This corresponds to equivariantly attaching to
the space obtained in the previous paragraph 
the above neighbourhood
$\KD^2 \times_{\KZ_2} S^1$  of a special exceptional 
orbit.

By Ref.\ \cite{Raymond}, $M_{g,h,t}$ is isomorphic to an equivariant 
connected sum of a principal $S^1$-bundle over a $2$-manifold
with boundary together with  $t$ copies of 
$\KRP^2 \times S^1.$ The $2$-manifold with boundary may be written as an
equivariant connected sum of $3$-sphere with handles 
\\

\begin{gather}
M_{g,h,t} \simeq S^3 \#  (S^2 \times S^1)_1  \# \cdots  (S^2 \times S^1)_{2g + h -1} \#  (\KRP^2 \times S^1)_1 \# \cdots \# (\KRP^2 \times S^1)_t.
\nonumber
\end{gather}
 
If, in addition to the above, we have exceptional orbits we
label these by $I \subset \KZ^{+}.$ (If the three-manifold
is compact, $|I|$ is finite.) Let $(\mu_i, \nu_i)$ be the
invariants of the exceptional orbit associated to $i \in I.$

Recall that a Lens Space $L(p,q)$ is a quotient of a three-sphere
$S^3 = \{ (z,w) | |z|^2 + |w|^2 = 1 \}$ by an action of $\KZ_p,$ 
$\mu \cdot (z,w) = (\mu z, \mu^{q} w).$ 
By writing the three-sphere as a glue of two solid tori, we may
(see Ref.\ \cite{Raymond}, Sec.\ (7) for details)
view $L(p,q)$ as obtained by gluing two solid tori together.
The natural circle action on $S^3$ given by the Hopf fibration
descends to an action of the circle group on $L(p,q).$
Due to the above decomposition of $L(p,q)$ as a gluing of two
solid tori, we obtain an equivariant decomposition of $L(p,q)$
as a glue of two solid tori one with a circle group action as in Eq.\ (\ref{FPNbd}) above
and one with a circle group action as in Eq.\ (\ref{ExNbd}) above.
Thus, the circle group action on $L(p,q)$ has a circle of fixed
points and one exceptional orbit with invariant $(p,q).$

We pick a toroidal neighbourhood of each exceptional orbit such that
the $S^1$-action on the three-manifold restricted to this
neighbourhood is equivariantly homeomorphic to the given toroidal
neighbourhood of the exceptional orbit in a Lens
Space $L(\mu,\nu)$ where $(\mu,\nu)$ are the invariants of the 
exceptional orbit described in Ref.\ \cite{Raymond}, Sec.\ (5), Lemma (5).
Let the $i$'th Lens Space in the above have
invariants $\mu_i, \nu_i.$ Then Ref.\ \cite{Raymond}, Thm.\ (4) shows
that $M$ is the equivariant connected 
sum
\begin{gather}
M \simeq M_{g,h,t}\# L(\mu_1,\nu_1) \# L(\mu_2,\nu_2) 
\# \cdots \# L(\mu_k,\nu_k).
\label{MESEFix}
\end{gather}

We need to determine the T-dual of $L(p,q)$ above. 
First we determine the T-dual of a solid torus with the
action of Eq.\ (\ref{ExNbd}) on it. 
\begin{lemma}
Consider the solid torus $\KD^2 \times S^1$ with
the circle action
$$
z \cdot (\rho e^{i \theta}, e^{i \psi}) = (z^{\nu} \rho e^{i \theta}, z^{\mu} e^{i \psi}),
$$
where $z \in S^1,$ $\nu, \mu$ are relatively prime and $0 < \nu < \mu.$ (This 
is the circle action in Eq.\ (\ref{ExNbd}) above.)
The T-dual of the above space with any $H-$flux $H$ 
may be calculated using the methods
of Ref.\ \cite{Bunke1}.
\label{LemExNbd}
\end{lemma}
\begin{proof}
We divide the proof into two cases:
Consider the solid torus $\KD^2 \times S^1$ with
the circle action in Eq.\ (\ref{ExNbd}) above
$$
z \cdot (\rho e^{i \theta}, e^{i \psi}) = ( z^{\nu} \rho e^{i \theta}, z^{\mu} e^{i \psi}),
$$
where $z \in S^1,$ $\nu, \mu$ are relatively prime and $0 < \nu < \mu.$
This is  the Lens Space $L(p,q)$ above. We argue that $L(p,q)$
is the total space of a  principal circle bundle $E$ 
over a orbifold $B.$ Thus it may be T-dualized by
the methods of Ref.\ \cite{Bunke1}, Sec.\ (5.2).

We consider the disk with the $\KZ_{\mu}$ action
which rotates each point on the disk by $2\pi/\mu.$
We consider the orbifold quotient $B$ of the disk by
the given $\KZ_{|\mu|} $action. Note that $B$ is homeomorphic 
to a slice around the exceptional orbit $\{ 0 \} \times S^1 \subseteq L(p,q).$
Hence  the Lens Space $L(p,q)$ above is the total space of a circle bundle
over $B.$

Following Ref.\ \cite{Bunke1}, Sec. (5.2),
we view the base orbifold $B$ above as a groupoid
quotient $B = [\G^1/G^0], $ where the object space is
$G^0 := \tilde{\bar{U}}:$ Here, $\tilde{\bar{U}}$ is
an orbifold chart of the disk and we require that
$\tilde{\bar{U}} \to \bar{U}$ is the $n-$fold covering map
$z \to z^{\mu}$ of the disc $U$ in complex plane. 
We require that the following
diagram commutes
$$
\begin{CD}
\tilde{\bar{U}} @>>f> \bar{U}\\
@VV{\simeq}V             @VV{\simeq}V \\
\tilde{\KD}^2 \subseteq \KC@>> {z \mapsto z^{|\mu|}} > \KD^2 \subseteq \KC.
\end{CD}
$$

Note that there is a natural action of $\KZ_{|\mu|}$ on $\tilde{\KD^2}$ which
permutes the fiber of the map $z \mapsto z^{|\mu|}.$ This lifts to a natural
action on $\tilde{\bar{U}}$ in the above diagram.
We view $\tilde{\bar{U}}$ as an {\em orbifold chart} of the base orbifold $B.$
We fix the arrow space $\G^1$ by requiring that groupoid $\G$ is the 
transformation groupoid $\KZ_{|\mu} \times \tilde{\bar{U}} \Rightarrow \tilde{\bar{U}}$
of the of the natural $\KZ_{|\mu|}$ action on $\tilde{\bar{U}}.$

We view $E$ a groupid quotient $[\E/\G^1]$ where $\E \to \G$ is an
equivariant $S^1-$bundle of groupoids. $\E$ is specified by an equivariant
$S^1-$bundle $\E \to \G^0, $ together with an action $\G^1 \times_{G^0} \E \to \E.$
We set $\E := S^1 \times \G^0$ and we let $Z_{|\mu|}$ act on the fiber over
any point in $\tilde{\bar{U}}$ with character 
$\chi([q]) = \exp(2 \pi i \frac{ p q}{|\mu|}), [q] \in Z_{|\mu|}.$ 
 
By Ref.\ \cite{Raymond}, (Sec.\ (5), after Eq.\ (5.1)), this is exactly the circle action on
the torus neighbourhood of the exceptional orbit in Eq.\ (\ref{ExNbd}) above.


We have a pullback square
$$
\begin{CD}
E @>>> [U(1)/\KZ_{|\mu|}]\\
@VVpV    @VVV \\
B @>>> [*/\KZ_{|\mu|}]\\
\end{CD}
$$
where the horizontal maps are equivariant homotopy equivalences. The vertical
maps are the bundle projections. ( The principal bundle in the second column is
the one in Ref.\ \cite{Bunke1}, Sec.\ (5.1).)
The stack cohomology of $B$ now follows from the fact that $H^*([*/\KZ_{|\mu|}], \KZ) \simeq H^{\ast}(B\KZ_{|\mu|},\KZ).$
By the proof of Ref.\ \cite{Bunke1}, Sec.\ (5.1): $H^0(B,\KZ) = \KZ,$ $H^{2l-1}(B,\KZ) = 0$ and $H^{2l}(B,\KZ) = \KZ/{|\mu|} \KZ$
for $l=1,2,\ldots.$
Further,  $c_1(E) = [q].$ 

Also, the Gysin sequence for $E \to B$ is the same as the Gysin sequence for
the principal bundle of stacks $[U(1)/\KZ_{|\mu|}] \to [*/\KZ_{|\mu|}]$ calculated
in Ref.\ \cite{Bunke1}, Sec.\ (5.1). This Gysin sequence is 
$$
0 \to H^3(E,\KZ) \overset{\pi_!}{\to} \KZ/{|\mu|}\KZ \overset{[q]}{\to} \KZ/{|\mu|}\KZ \to \ldots
$$
and by exactness, $\pi_!$ is injective. Hence $H^3(E,\KZ)$ is isomorphic to the kernel
of the cup product of a class in $\KZ/{|\mu|}\KZ$ with $[q].$ Thus,
$$H^3(E,\KZ) \simeq \{ [s] \in \KZ/{|\mu|}\KZ \mbox{ such that } \mu \mbox{ divides } sq \} \subset \KZ/{|\mu|}\KZ.$$
We pick a $H-$flux $[s]$ in this group.

We note that the collapse map $B \to [*/\KZ_{|\mu|}]$ is a homotopy equivalence.
By Lemma (\ref{LemTTDWHE}) above, this implies that there is map from the
T-dual of $E$ to the T-dual of $[U(1)/\KZ_{|\mu|}].$

Since $E$ is a pullback, if we pullback to a small contractible open subset of the base
which does not contain the singular point at $(0,0),$ it is clear that the bundle 
over this set must be trivial.  Hence, the T-dual of $E \to B$ restricted
to a small open set in $B$ which does not contain the singular point at $(0,0),$
is trivial.  By the semi-locality of Topological T-duality, this implies that the 
T-dual is a trivial circle bundle over $B$ {\em away from } $(0,0).$
The topology of the T-dual is only from the topology of the fiber over the
singular point at $(0,0).$

Suppose $(|\mu|, q) = 1.$ By Ref.\ \cite{Bunke1}, this implies that
the $H-$flux on $E$ is zero. 
 By Ref.\ \cite{Bunke1}, the T-dual is $[U(1)/\KZ_{|\mu|}]$
with a {\em noneffective} $\KZ_{|\mu|}-$action. 

In this situation, the T-dual of $E$ is $\KD^2 \times S^1$ with the translation 
$S^1-$action away from the singular fiber $S^1 \times \{0\}.$  
The fiber at $(0,0)$ is $[U(1)/\KZ_{|\mu|}]$
with a noneffective $\KZ_{|\mu|}$ action. As in the last paragraph of
Ref.\ \cite{Bunke1}, Sec.\ (5.1)) this is not equivalent to a space.

We may argue similarly for $h \neq 0.$ The topology of the T-dual in this
case is obtained from the Gysin sequence for the T-dual bundle $\hat{p}:\hat{E} \to B$
in Ref.\ \cite{Bunke1}, Sec.\ (5.1.6). By Lemma (\ref{LemTTDWHE}) above, the T-dual
of the  bundle $p$ is pulled back from the T-dual $\hat{q}$ of the bundle 
$q:[U(1)/\KZ_{|\mu|}] \to [*/\KZ_{|\mu|}]$ and every pair on $\hat{p}$ is pulled back from
a pair on $\hat{q}.$ 

By Ref.\ \cite{Bunke}, Lemma\ (2.1.2), we have $c_1(\hat{E}) = -\pi_!(h).$
Thus, $\hat{p}: \hat{E} \to B$ with the above $H-$flux has
characteristic class $c_1(\hat{E}) = [-s] \in \KZ/{|\mu|}\KZ$ since the $H-$flux
on $p:E \to B$ (see above) was $[s].$ 

The Gysin sequence for the T-dual is
$$
0 \to H^3(\hat{E},\KZ) \overset{\hat{\pi}_!}{\to} \KZ/{|\mu|}\KZ \overset{[-s]}{\to} \KZ/{|\mu|}\KZ \to \ldots,
$$
By exactness, $\hat{\pi}_!$ is injective. Hence $H^3(\hat{E},\KZ)$ 
is isomorphic to the kernel of the cup product of a class in 
$\KZ/{|\mu|}\KZ$ with $[-s].$ Thus, it is isomorphic to
$\{ [r] \in \KZ/{|\mu|}\KZ | \mu | sr \} \subset \KZ/{|\mu| \KZ}.$
From the Gysin sequence for the T-dual bundle, 
$c_1(E) = [q] = -{\hat{\pi}}_!(\hat{h}).$ However,
by the above $\hat{\pi}_!$ is injective. Hence,
the T-dual $H-$flux is $\hat{h} = [-q].$ 

The T-dual is a $S^1-$stack with the above topology, characteristic
class and $H-$flux.

\end{proof}

Then we determine the T-dual of $L(p,q)$ above.
\begin{lemma}
The T-dual of $L(p,q)$ with any $H$-flux can be calculated
using the above. 
\label{LemTTDLPQ}
\end{lemma}
\begin{proof}

We divide the proof into two cases:
First, from Ref.\ \cite{Raymond}, Secs.\ (7,8),
$L(p,q)$ is obtained by equivariantly gluing
a solid torus with the circle action equivariantly homeomorphic
to Eq.\ (\ref{ExNbd}) above to a solid torus with circle action
equivariantly homeomorphic to Eq.\ (\ref{FPNbd}) above. The boundary
two-tori of each of the solid tori above are glued using a map
from $\KT^2 \to \KT^2$ of degree $(p,q)$ above.

We use the convention of Ref.\ \cite{Raymond} to describe these spaces.
The spaces above are viewed as $S^1$-spaces with quotient space
susbsets the unit disk in the plane $\KD^2 \subseteq \KR^2.$
The first space is a semi-free space over an annulus
$[1,1/2] \times S^1$ with $\{ 1 \} \times S^1$ a circle of fixed points of
the circle action and every other circle orbit free. The circle group action
on this space is the action in Eq.\ (\ref{FPNbd}) above.

This space in the previous paragraph is glued to a principal circle
bundle over a disk orbifold, i.e. a space of the form 
$\KD^2/\KZ_{|\mu|}.$ Here, the disk of radius $\{ 1/2 \}$ in the plane
is taken to be $\KD^2/\KZ_{|\mu|}$ by an orbifold chart and there is
an orbifold point at $(0,0).$ The circle action on the total space of the
bundle is the action in Eq.\ (\ref{ExNbd}) above. 

The  two spaces above are glued to each other along
the common torus boundary which is the set
$\KT^2 \to \{ 1/2 \} \times S^1$
by a map $f:\KT^2 \to \KT^2$ of degree $2.$
\leavevmode
\begin{enumerate}
\item{{\bf Zero $H-$flux:}}
By Thm.\ (\ref{ThmEAttach}) the T-dual of the above is the T-dual
of each of the two solid tori above glued together by a T-dual map.

The circle action in Eq.\ (\ref{FPNbd}) above has an annulus as
a quotient space. The T-dual 
of the action in Eq.\ (\ref{FPNbd}) is a trivial circle bundle over 
the annulus $[1,1/2] \times S^1$ with
a source of $H-$flux at $ \{1\} \times S^1$
This is because the circle action is semi-free.

The T-dual of the action in Eq.\ (\ref{ExNbd}) was calculated in
the Lemma\ (\ref{LemExNbd}) above. 
By that Lemma, the T-dual with zero $H-$flux is $\KD_{1/2}^2 \times S^1$ with
the fiber $S^1 \times \{ 0 \}$  equivariantly $2-$equivalent to $[U(1)/\KZ_{|\mu|}]$ with
a {\em noneffective } $\KZ_{|\mu|}-$action.
The above two spaces are glued together along  the common boundary
$\{ 1/2 \} \times S^1 \times S^1.$

This implies that  the T-dual is the gluing of the above two spaces together
on the $\KT^2$ which is their common boundary by 
a map $f:\KT^2 \to \KT^2.$ 
We  argue that this map must be nullhomotopic.

\item{{\bf Nonzero $H-$flux:}} By Lemma\ (\ref{LemExNbd}) above, the
T-dual of the above space with $H-$flux is a $S^1-$bundle of stacks
$p:\hat{E} \to B$ with $H-$flux.

This bundle has characteristic class $c_1(\hat{E}) = [-s] \in \KZ/{|\mu|}\KZ$
since the $H-$flux on $p:E \to B$ (see above) was $[s].$ 

Hence $H^3(\hat{E},\KZ)$ is isomorphic to
$\{ [r] \in \KZ/{|\mu|}\KZ | \mu | sr \} \subset \KZ/{|\mu| \KZ}.$
From the Gysin sequence for the T-dual bundle, 
$c_1(E) = [q] = -{\hat{\pi}}_!(\hat{h}).$ However,
by the above $\hat{\pi}_!$ is injective. Hence,
the T-dual $H-$flux is $\hat{h} = [-q].$
The T-dual is the $S^1-$bundle of stacks with the above topology,
characteristic class and $H-$flux.
\end{enumerate}
\end{proof}
The above facts together with Cor.\ (\ref{CorEquivSum}) above,  let us
determine the T-dual of any $S^1$-action on a three-manifold without
$H$-flux.

\begin{theorem}
Let $W$ be a compact $3$-manifold with
smooth circle action. 
\leavevmode
\begin{enumerate}
\item The Topological T-dual of the stack
$\underline{W}$ without $H$-flux  is an equivariant connected sum of some number copies of  T-duals
of the following spaces
with a particular gluing. (The $S^1$-actions
on these spaces are given in the proof below and in Eqs.\ (\ref{FPNbd}, \ref{ExNbd}) above.)
\begin{enumerate}
\item $S^3,$
\item $S^2 \times S^1,$
\item $L(\mu,\nu),$
\item $\KRP^2 \times S^1,$
\item $\KD^2 \times S^1$ with the action in Eq.\ (\ref{FPNbd}) above.
\end{enumerate}
The T-duals of these spaces are listed below.
\item The Topological T-dual of the stack $\underline{W}$ with 
$H$-flux $\mathcal{H}$ is an equivariant connected sum of the T-duals of
the spaces in the previous item {\em with} $H-$flux. These T-duals
are listed in the proof below.
\end{enumerate}
\label{ThmTD3MfD}
\end{theorem}
\begin{proof}
\leavevmode
\begin{enumerate}
%
%
\item By the above, $W$ is 
equivariantly homeomorphic to an equivariant connected sum
$W \simeq M_{g,h,t} \# \L(\mu_1,\nu_1) \# \ldots \# \L(\mu_k,\nu_k)$ 
where $M_{g,h,t}$ was defined above.

By Cor.\ (\ref{CorEquivSum}) above, the T-dual of $W$ is the connected sum of the T-duals
of $M_{g,h,t}$ with the T-duals of the $L(\mu_i,\nu_i).$ However, by the
above $M_{g,h,t} \simeq S^3 \# (S^2 \times S^1) \# \dots \# (S^2 \times S^1) \#( \KRP^2 \times S^1) \# \dots \# (\KRP^2  \times S^1).$

The T-duals of each of the stacks in the above connected sum without $H$-flux are well known, since they
are actually the total spaces of principal bundles, usually trivial ones. Since there is no $H-$flux,
the T-dual will be a trivial bundle. In partcular, the T-dual of $S^3$ will be $S^2 \times S^1$ with $H-$flux the generator
of $H^3(S^2 \times S^1,\KZ) = \KZ.$
None of the other spaces in the connected sum for $M_{g,h,t}$ will T-dualize. Hence the T-dual of $M_{g,h,t}$ is a connected
sum of copies of $S^2 \times S^1$ and $\KRP^2 \times S^1$ only with $H-$flux supported on some of the $S^2 \times S^1$
factors. The T-dual of $W$ is connected sum of the above with the T-dual of the lens spaces. Thus, the only problem in computing
the T-dual of a three-manifold with circle action is the computation of the T-dual of 
the lens spaces $L(\mu,\nu)$ with and without $H-$flux.

The T-dual of the Lens Space $L(\mu,\nu)$ without $H-$flux was calculated in
Lemma (\ref{LemTTDLPQ}) Item (1) above.
The T-dual consists of the T-dual of the solid torus with the circle action in Eq.\ (\ref{FPNbd}) 
above glued to the T-dual of the solid torus with the circle action in Eq.\ (\ref{ExNbd}) above.
Since there is no $H-$flux, the T-dual is the space described in Item (1) of that lemma and consists of a principal circle bundle of stacks whose base
is a disk orbifold whose orbifold chart is the closed unit disk $\KD^2 \subseteq \KR^2.$ The total space
is obtained by gluing a trivial bundle over the annulus $[1,1/2] \times S^1$ with $H-$flux supported at $\{ 1 \} \times S^1$
(and no other $H-$flux) to a trivial circle bundle over the disk of radius $1/2$ around $(0,0) \in \KR^2$ with the fiber over $(0,0)$
possessing a {\em noneffective} $\KZ_{|\mu|}-$action. The point $(0,0)$ is an orbifold point.
\item By the above, $W$ is equivariantly homeomorphic to an equivariant connected sum
$$
W \simeq M_{g,h,t} \# \L(\mu_1,\nu_1)\# \ldots \# \L(\mu_k,\nu_k)
$$
where 
$$
M_{g,h,t} \simeq S^3 \# (S^2 \times S^1) \# \dots \# (S^2 \times S^1) \# (\KRP^2 \times S^1) \# \dots \# (\KRP^2 \times S^1)$$ 
as defined above.

Let $H_i$ be the restriction of $H$ to the $i'$th factor of the above equivariant connected sum.
By the above, the T-dual of $W$  with $H-$ flux is the equivariant connected sum of the T-duals of each of the
above spaces with the restriction of the $H-$flux.
The T-dual of each of the above factors is well known (see Ref.\ \cite{BEM}) except the T-dual of
$L(\mu,\nu)$ which was worked out in Lemma (\ref{LemTTDLPQ}) above. Thus, the T-dual of $W$
is a connected sum of the following spaces:
\begin{itemize}
\item {\bf The T-dual of $p:S^3 \to S^2$ with $H-$flux $H_1:$} By Ref.\ \cite{BEM}, the T-dual is 
the principal circle bundle $q:S^3/\KZ_k \to S^2$ with $1$ unit of $H-$flux where $k$ is the value
of $p_!(H_1) \in \KZ \simeq H^2(S^2,\KZ).$
\item {\bf The T-dual of $p:S^2 \times S^1$ with $H-$flux $H_2:$} By Ref.\ \cite{BEM}, the T-dual is
the circle bundle $q:S^3/\KZ_k \to S^2$ with no $H-$flux where $k$ is the value of
$p_!(H_1) \in \KZ \simeq H^2(S^2,\KZ).$
\item {\bf The T-dual of $\KRP^2 \times S^1 \to \KRP^2$ with $H-$flux $H_3:$} The cohomology of $\KRP^2 \times S^1$
is $H^0(\KRP^2 \times S^1,\KZ) = \KZ,$ and $H^3(\KRP^2 \times S^1,\KZ) = \KZ_2.$ Hence, $H_3$ can only
be zero or nonzero. If the $H-$flux is zero, there is no T-duality. If the $H-$flux is nonzero ($1 \in \KZ_2$), the
result is  the only nontrivial bundle $E$ over $\KRP^2,$ (recall $H^2(\KRP^2,\KZ) \simeq \KZ_2$) $q:E \to \KRP^2.$
This bundle was calculated in Ref.\ \cite{BEM}, Sec.\ (4.4). The cohomology of the total space of the bundle is
$$
H^0(E,\KZ) = \KZ,
H^1(E,\KZ) = \KZ,
H^2(E, \KZ) = 0,
H^3(E,\KZ) = \KZ_2.
$$
and the resulting bundle has trivial $H-$flux.
\item {\bf The T-dual of $L(\mu,\nu)$ with $H-$flux $H_4:$} This was calculated in Lemma\ (\ref{LemTTDLPQ}) above.
By Lemma\ (\ref{LemExNbd}) above, the
T-dual of the above space with $H-$flux is a $S^1-$bundle of stacks
$p:\hat{E} \to B$ with $H-$flux.

This bundle has characteristic class $c_1(\hat{E}) = [-s] \in \KZ/{|\mu|}\KZ$
since the $H-$flux on $p:E \to B$ (see above) was $[s].$ 

Hence $H^3(\hat{E},\KZ)$ is isomorphic to
$\{ [r] \in \KZ/{|\mu|}\KZ | \mu | sr \} \subset \KZ/{|\mu| \KZ}.$
From the Gysin sequence for the T-dual bundle, 
$c_1(E) = [q] = -{\hat{\pi}}_!(\hat{h}).$ However,
by Lemma\ (\ref{LemTTDLPQ}) above, $\hat{\pi}_!$ is injective. Hence,
the T-dual $H-$flux is $\hat{h} = [-q].$
The T-dual is the $S^1-$bundle of stacks with the above topology,
characteristic class and $H-$flux.
\end{itemize}
\end{enumerate}
The above facts together with Cor.\ (\ref{CorEquivSum}) above,  let us
determine the T-dual of any $S^1$-action on a three-manifold without
$H$-flux.
\end{proof}

It would be interesting to relate the twisted $K$-theories of the two sides
of the duality. However, since the base is not a Deligne-Mumford stack,
the results of Bunke et al concerning the isomorphisms of twisted
$K-$theory in Ref.\ \cite{Bunke2} will not apply.
We would need the results of Mathai and Wu (see Ref.\ \cite{MaWu}, Appendix and
Thm.\ (1) for details). We can show using these that the twisted $K-$theory of
$W$ is isomorphic to the twisted equivariant $K-$theory of the corrrespondence
space. When the T-dual exists in the sense of Mathai-Wu, this is the twisted
equivariant $K-$theory of the T-dual. In addition, the twisted cohomology
of $W$ is isomorphic to the equivariant twisted cohomology of the correspondence
space. Similarly to the above, when the T-dual exists, this is the twisted cohomology
of the T-dual. 

Thus, if there are no exceptional orbits, the twisted $K-$theory and twisted cohomology
of $W$ are isomorphic (with a degree shift) to the twisted $K-$theory and twisted
cohomology of the T-dual.  

If there are exceptional orbits, the correspondence space remains a space, but the
T-dual spacel is in general a stack. 
If one could define twisted cohomology and twisted $K-$theory for stacks along the lines
of Ref.\ \cite{Ginot}, Def.\ (5.1), we could extend the above isomorphism to include this
case. 

In this section we assumed that $W$ was compact. The case of noncompact
$W$ can be handled, we might have to take the equivariant connected sum
of an infinite number of lens spaces as in Ref.\ \cite{Raymond}.

In this section we have calculated the Topological T-dual of an arbitrary
compact three-manifold with circle action.
It would be interesting to relate this to the known T-duals in String Theory.
Since the T-duals of $KK$-monopole spacetimes agree with the corresponding
T-duals in String Theory (see Ref.\ \cite{Pande} for details), one would expect
these T-duals to match as well.

There appear to be two ways to connect three-manifolds with
String Theory backgrounds: As was argued at the beginning of
this section, one could view the three-manifold as a spacelike slice of
a four manifold. One could then T-dualize the four-manifold by viewing
it as a String Theory background in the usual ways, for example by
viewing it as a compactification. It would be interesting to see if the
Topological T-dual matches with the String Theoretic T-dual in this case.

There might be another way to connect the Topological T-dual of
three-manifolds with string theory: We could view the three-manifold
as the event horizon of a higher-dimensional black hole (see Ref.\ \cite{Hollands}
for details).

It is well known (Hawking's Theorem) that a stationary
axisymmetric four dimensional spacetime always
possesses a $U(1)$ isometry. The fixed point
sets of the isometries are stratified by dimension.
The event horizon of the black hole is the highest dimension
fixed point set. For a black hole in a four-dimensional spacetime this is
two-dimensional.

A recent topic in String Theory is the study of higher-dimensional
black holes. Under certain conditions (see Ref.\ \cite{Hollands}),
it can be shown that there exist five-dimensional black hole spacetimes
with isometry groups containing the group $U(1) \times U(1).$
For such spacetimes it can proved that the event horizon is a three-manifold
with a circle action. The classification of three-manifolds with
circle actions has been used to study these spacetimes in Ref.\ \cite{HolIshi}
(see, for example, the results in Sec.\ (2) of that paper).
Such spacetimes are valid backgrounds for String Theory
possibly after adding D-branes. It would be interesting to see the
effect of T-duality on these spacetimes. Presumably, the Topological
T-dual of a three-manifold with circle action above would be related
to these T-duals.

An example of a higher dimensional black hole with the above properties
is the spinning black ring on Taub-NUT space. This space has a $U(1)-$isometry
along the isometry direction $x^5$ (in the notation of Ref.\ \cite{GSY}).
The geometry of this background in terms of D-branes
has been discussed in Ref.\ \cite{GSY}.

\section{Final Remarks}
\label{secFinal}
In this paper we have studied the Topological T-dual of spaces 
containing $KK$-monopoles. In Secs.\ (\ref{secComp},\ref{secTDKK}) of
the paper above, we have explicitly calculated the T-dual of several spaces
with $KK$-monopoles. In Sec.\ (\ref{secDyon}) we have attempted 
to model the `Dyonic Coordinate'
associated with $KK$-monopoles within the stack theory formalism.
In Ref.\ \cite{Pande} we had obtained a model for the Dyonic Coordinate
using the $C^{\ast}$-algebraic formalism of Topological T-duality.
It is interesting that the same phenomenon appears in two completely
independent approaches to Topological T-duality.

The formalism of Bunke et al.\ \cite{Bunke, Bunke1, Bunke2}
may be applied to a $KK$-monopole by passing to the 
associated simplicial bundle and 
taking the Topological T-dual of that simplicial bundle. 
When the above stack is the stack associated to the total 
space of a principal circle bundle, the formalism of 
Bunke and co workers in Ref.\ \cite{Bunke1} would give the same
answer as the  $C^{\ast}$-algebraic formalism. However, if the stack was the
stack associated to a space with a non-free circle action we show by 
an explicit example that the two formalisms dol not give the same 
T-dual. This is because the formalism of Bunke and coworkers
`regularizes' the neighbourhood of the
fixed point when it passes from the given stack 
to its associated simplicial bundle.

Also, as has been argued above, the $C^{\ast}$-algebraic T-dual
lifts the $S^1$-action on the space to a $\KR$-action on the
$C^{\ast}$-algebra. The formalism of Bunke et al.\ does not
lift the $S^1$-action to a $\KR$-action. This causes a difference
in T-duals when fixed points are encountered. The formalism of Mathai-Wu
\cite{MaWu} also does not lift the $S^1$-action to a $\KR$-action.

In addition, a connection between the topology of the T-dual and the
coarse moduli space of the T-dual stack was demonstrated in 
Sec.\ (\ref{secTDKK}) after Cor.\ (\ref{CorKKMulti}) . There we
observed that in the example of the $KK-$monopole, the physical
T-dual spacetime was the coarse moduli space of the T-dual stack.
It would be interesting to investigate this further. In particular, if the
coarse moduli space of the T-dual is different from the T-dual stack, the
T-dual stack possesses nontrivial inertia groups.  This appearance
of the coarse moduli space is extremely interesting, and seems to
be connected to the nontrivial stabilizers of the non-free orbits of
the original stack.

In Sec.\ (\ref{secClass}) we argued that the methods of Bunke et al and
Mathai and Wu were connected by the notion of a classifying
space of a stack: They differed only in the choice of a classifying
space. It would be interesting to see if a new prescriptions
for Topological T-duality similar to the formalisms of
Bunke et al and Mathai and Wu could be made by other
natural choices for a classifying space for other families of stacks
apart from the ones in this section. Since the notion of a classifying
space for a topological stack has been studied in detail (see
Refs.\ \cite{Noohi3}) there is a chance that this might
be possible.

Under certain conditions on the stack quotient map $p:\E \to \Y$
we dervie a stack cohomology Gysin sequence. 
We then prove that this Gysin 
sequence may also be used to determine the T-dual of a principal
circle bundle of stacks.

It is interesting to note that the calculation of Topological T-duals
for $U(1)$-gerbes on a principal $\KT^n$-bundle over an arbitrary
topological groupoid using the theory of crossed products of
groupoid $C^{\ast}$-algebras has also been done by Daenzer in 
Ref.\ \cite{Daenzer}. The formalism can T-dualize non-free
group actions. It would be interesting to compare the results of
that formalism with the results of this paper.

In Sec.\ (\ref{sec3Mfd}) we discuss the T-dual of an arbitrary
compact three-manifold without boundary with a circle action.
If there are no exceptional orbits, with zero $H-$flux,
the T-dual is a connected sum of $S^2 \times S^1$ and
$\KRP^2 \times S^1$ with $H-$flux supported on some
of the $S^2 \times S^1$ factors. If there are no exceptional
orbits, with nonzero $H-$flux, the T-dual is copies of $S^3/\KZ_k$
with $H-$flux glued to copies $S^3/\KZ_m$ without $H-$flux 
and to $E \to \KRP^2$ (the only nontrivial bundle over $\KRP^2$).
If there are exceptional orbits present, however, we have to
glue in copies of the stacks described in Thm.\ (\ref{ThmTD3MfD}) above,
one for each exceptional orbit.

Some connections between string theory backgrounds and the above
examples were noted in this section as well. Among other things, 
it was noted that the event horizon
of a five dimensional black hole background with a 
$U(1) \times U(1)-$isometry 
is a three-manifold with $U(1)-$isometry. It would be interesting
if Topological T-duality was able to describe the T-dual of these 
String/M-Theory backgrounds.

\section*{Acknowledgements}
\flushleft{I thank Professor Jonathan M. Rosenberg of the Department of
Mathematics,  University 
of Maryland, College Park for all his advice and help.}
\par
\flushleft{I thank the Mathematics Department, Harish-Chandra Research
Institute, Allahabad, for support in the form of a postdoctoral fellowship
during the writing of this paper.}
\par
\flushleft{I thank the School of Mathematical Sciences, NISER, Bhubaneshwar,
for support in the form of a Visiting Position during the writing of 
this paper.}
\flushleft{I thank Professor Calder Daenzer, Department of Mathematics,
Penn State University for useful comments.}

\end{document}